\documentclass[11pt, a4paper]{article}
\pdfoutput=1 %

\usepackage{fullpage}

\usepackage[utf8]{inputenc}
\usepackage[T1]{fontenc}

\usepackage{libertine}

\newif\ifshortJournalName
\shortJournalNamefalse
\newif\ifshortConferenceName
\shortConferenceNamefalse

\makeatletter
\newcommand{\formatConferenceAbbr}[2]{{#1}~'\@gobbletwo#2}

\makeatother

\usepackage{amsfonts}
\usepackage{amsmath}
\usepackage{amssymb}
\usepackage{stmaryrd}
\usepackage{mathtools}
\usepackage{nicefrac}
\usepackage{bm}
\usepackage{mathrsfs}
\usepackage{amsthm}

\usepackage{xspace}
\usepackage{orcidlink}
\usepackage{authblk}

\usepackage{lineno}

\usepackage[inline]{enumitem}

\usepackage{xcolor}
\usepackage{tikz}
\usetikzlibrary{backgrounds,calc,patterns,nfold,positioning}
\usepackage{pgfplots}
\usepackage{subcaption}
\pgfplotsset{compat=1.18}
\usepackage{float}
\usepackage{wrapfig}

 \usepackage{hyperref}
\hypersetup{
    pdfencoding=auto, 
    psdextra,
    colorlinks=true,
    citecolor=green!40!black,
    linkcolor=red!50!black,
    urlcolor=blue!80!black
}

\usepackage{cleveref}
\usepackage[sort,compress]{cite}
\usepackage{doi}
\usepackage{natbib}

\usepackage{tabularx}
\usepackage{booktabs}
\usepackage{multirow}

\newcommand{\N}{\mathbb{N}}

\newcommand{\R}{\mathbb{R}}

\newcommand{\Yes}{\textsc{Yes}\xspace}

\newcommand{\agents}{{\ensuremath{{N}}}}    %
\newcommand{\numAgents}{{\ensuremath{{n}}}} %

\newcommand{\alloc}{\pi}

\DeclareMathOperator{\util}{u}
\DeclareMathOperator{\USW}{USW}
\DeclareMathOperator{\NSW}{NSW}

\DeclareMathOperator{\dist}{dist}

\newcommand{\probName}[1]{\textsc{#1}\xspace}

\newcommand{\Oh}[1]{{\mathcal{O}\mathopen{}\left(#1\right)}}
\newcommand{\oh}[1]{{o\mathopen{}\left(#1\right)}}

\NewDocumentCommand{\cc}{ O{} O{} m }{\mbox{%
    \expandafter\ifx\expandafter\relax\detokenize{#2}\relax\else{#2-}\fi%
    \textsf{#3}%
    \expandafter\ifx\expandafter\relax\detokenize{#1}\relax\else{-#1}\fi%
    }\xspace}
\newcommand{\DeclareComplexityLowerBound}[1]{%
  \expandafter\newcommand\csname #1\endcsname{\cc{#1}}%
  \expandafter\newcommand\csname #1h\endcsname{\cc[hard]{#1}}%
  \expandafter\newcommand\csname #1hness\endcsname{\cc[hardness]{#1}}%
  \expandafter\newcommand\csname #1c\endcsname{\cc[complete]{#1}}%
  \expandafter\newcommand\csname #1cness\endcsname{\cc[completeness]{#1}}%
}
\newcommand{\DeclareComplexityLowerBoundParam}[2]{%
  \expandafter\newcommand\csname #1\endcsname[1][#2]{\cc{#1[##1]}}%
  \expandafter\newcommand\csname #1h\endcsname[1][#2]{\cc[hard]{#1[##1]}}%
  \expandafter\newcommand\csname #1hness\endcsname[1][#2]{\cc[hardness]{#1[##1]}}%
  \expandafter\newcommand\csname #1c\endcsname[1][#2]{\cc[complete]{#1[##1]}}%
  \expandafter\newcommand\csname #1cness\endcsname[1][#2]{\cc[completeness]{#1[##1]}}%
}
\newcommand{\DeclareComplexityUpperBound}[1]{
  \expandafter\newcommand\csname #1\endcsname{\cc{#1}}%
}

\DeclareComplexityLowerBound{NP}
\DeclareComplexityLowerBound{paraNP}
\DeclareComplexityLowerBound{coNP}
\DeclareComplexityLowerBoundParam{W}{1}
\DeclareComplexityLowerBoundParam{coW}{1}
\DeclareComplexityUpperBound{FPT}
\DeclareComplexityUpperBound{XP}

\newtheorem{theorem}{Theorem}

\newtheorem{corollary}{Corollary}
\newtheorem{proposition}{Proposition}

\newtheorem{remark}{Remark}
\newtheorem{observation}{Observation}%
\Crefname{observation}{Observation}{Observations}
\newtheorem{claim}{Claim}%
\Crefname{claim}{Claim}{Claims}

\newenvironment{claimproof}[1]{\emph{Proof.}\hspace{0.15cm}#1}{\hfill$\blacktriangleleft$\medskip}
\newcommand{\proofparagraph}[1]{\smallskip\emph{#1}\hspace{0.15cm}}

\usepackage{pict2e}
\newcommand{\opentriangle}{%
  \raisebox{0.2pt}{\makebox[0.77778em]{%
    \setlength{\unitlength}{0.6em}%
    \linethickness{0.4pt}%
    \begin{picture}(1,1)
    \polygon(0,0)(1,0)(1,1)
    \end{picture}%
  }}%
}

\definecolor{cbGood}{HTML}{1A9850} %
\definecolor{cbHard}{HTML}{D73027} %
\definecolor{cbOpen}{HTML}{E69F00}
\title{Network Allocation Games with Anonymous Preferences\footnote{An extended abstract of this work has been published in the proceedings of the 19th International Symposium on Algorithmic Game Theory, SAGT '26~\citep{DeligkasESV2026}. Compared to the conference version, this article contains full proofs of all results, several of which were omitted or only sketched due to space constraints. Moreover, \Cref{thm:beta-hurting:Non-existence}, the \W-membership part of \Cref{thm:USW:NPc}, and \Cref{cor:PO:verify:ETH} are new.}}

\author[1]{Argyrios Deligkas}
\author[1]{Eduard Eiben}
\author[2,3]{Šimon Schierreich}
\author[4]{Alexandros A. Voudouris}

\affil[1]{Royal Holloway, University of London, United Kingdom}
\affil[2]{AGH University of Krakow, Poland}
\affil[3]{Czech Technical University in Prague, Czechia}
\affil[4]{University of Southern Denmark, Denmark}

\begin{document}
\maketitle              %
\begin{abstract}
    We study network allocation games in which a set of agents must be allocated to (a subset of) the vertices of a graph topology.
    The agents are anonymous and strategic: Each of them aims to maximize her utility, which is a function of the number of agents in the neighborhood of their allocated vertex. 
    We focus on the existence and the computation of stable allocations under various notions of stability. 
    More specifically, we study \textit{swap-based} stability notions, where agents can exchange locations between them, and \textit{jump-based} stability notions, where agents can jump to an empty vertex of the graph.
    For almost all stability notions we consider, we provide dichotomies between existence and intractability
    that depend on the structure of the utility functions of the agents and the topology.
    In addition, we prove strong intractability results for computing welfare-maximizing allocations and for verifying whether a given allocation is Pareto optimal.
    
\end{abstract}

\section{Introduction}\label{sec:introduction}

Imagine your university has just opened a new campus aiming to expand, and you, as an already senior member of your department, have been assigned the important task of deciding how to allocate your newly-recruited colleagues in the available offices. 
Interestingly, the building in the new campus is a refurbished old factory, and the office space is quite unusual and does not resemble a grid; instead, trying to optimize the number of people it can fit, the space was partitioned like a pretty arbitrary graph. 
In addition, you do not really know your new colleagues, and, maybe with a few exceptions, they do not know each other either. So the only information you can get about their preferences is about how comfortable they would be with how many people have neighboring offices, or even are in shared rooms. In other words, they express anonymous preferences that assign values depending on the number of their neighbors. 
Your task then boils down to assigning offices such that the resulting allocation is ``stable'', i.e., none of your colleagues has any incentive to request a change. 

This toy, and, of course, completely imaginary example, leads to an interesting network allocation game, which, to the best of our knowledge, has not been studied in the past. Beyond office allocation, similar considerations arise in many other settings, such as assigning jobs to the interconnected machines in a computing cluster, where a job's performance depends on the number of busy neighboring machines.
More formally, in the game, we have a set of agents that must be allocated to (a subset of) the vertices of a graph. 
Crucially, though, the agents are {\em anonymous} and {\em strategic}. Each of them aims to maximize her utility, which is a function of the {\em number} of agents in the neighborhood of her allocated vertex, but {\em not} of their identity. 
Then, the goal is to find a {\em stable allocation} of the agents to the vertices, if such an allocation exists. Put simply, we want an allocation such that no agent can improve by (ex)changing her position.

This network allocation game shares common characteristics with two important classes of games that have attracted significant attention, that of {\em Hedonic games} (see, e.g., \citep{BanerjeeKS2001,BogomolnaiaJ2002,Ballester2004}) and that of {\em Schelling games} (e.g., see \citep{schelling-journal,chauhan2018schelling}). 
In Hedonic games, the objective is to compute a partition of the agents into groups that is stable with respect to the preferences of the agents over the members of their groups. 
In Schelling games, the agents have types, and the objective is to compute an allocation of them to the vertices of a graph that is stable, given that each agent aims to optimize a function of the number of agents of the same type in their neighborhood. 

In both cases, two main stability notions have been studied: {\em swap-stability}, where no pair of agents can mutually benefit by exchanging their locations, and {\em jump-stability}, where no agent has incentives to unilaterally move to another group or empty vertex (possibly under some restrictions).

Despite their similarities, neither of these models fully captures the one we study in this work. In Hedonic games, there is no underlying topology, whereas in Schelling games, the preferences of the agents are not anonymous. 
By embedding anonymous hedonic preferences into a graph structure, our model captures both the richness of network topology and the simplicity of aggregate-based preferences, leading to our question of interest: 

\begin{quote}
   \em  What type of stable outcomes exist for anonymous network allocation games, and how difficult is it to find them?
\end{quote}

\subsection{Our Contribution}\label{sec:contribution}

We study the existence and complexity of finding stable outcomes in anonymous network allocation games with respect to three dimensions: a) the stability notion; b) the structure of the utility functions of the agents; c) the structure of the topology graph. 
For the majority of possible combinations of the restrictions mentioned above, we provide dichotomies with respect to the existence of stable allocations {\em and} the complexity of computing them: either the stability concept always exists and can be computed in polynomial time, or it does not always exist and deciding whether an instance admits a solution is an intractable problem.

We begin by studying various notions of swap-based stability. We first prove that swap-stable allocations (in which no pair of agents wants to exchange their locations) always exist, independently of the topology graph and the utility functions of the agents, and can be found in quadratic time, with respect to the number of agents, via natural swap dynamics (\Cref{thm:swap_stable}). In fact, the same proof shows the same result even for {\em weak} swap-stable allocations, where at least one agent strictly increases her utility from the swap. 
We also consider envy-free allocations (in which no agent prefers the vertex allocated to a different agent more than her own); envy-freeness is the most prominent fairness notion and is stronger than swap-stability since it requires that all agents are completely satisfied with the allocation. Because of this, besides the notable case of two agents, we show that envy-free allocations do not always exist when there are at least three agents (\Cref{thm:alpha-envy-free:Non-existence}), and it is, in general, \NPc to decide their existence (\Cref{thm:alpha-envy-free:Hardness}). 
In fact, we show that our \NPcness result holds for two other cases: approximately envy-free allocations, and~$\beta$-hurting swap-stable allocations, a notion that interpolates between weak swap-stability and envy freeness. Our hardness result holds for any non-trivial approximation in both cases, even for restricted utility functions. Moreover, we show that for every~$\beta\in(0,1)$, a~$\beta$-hurting swap stable allocation may fail to exist, already with three agents (\Cref{thm:beta-hurting:Non-existence}).
On the positive side, we prove that both problems are \FPT parameterized by the number of agents (\Cref{thm:alpha-envy-free:FPT}).

We next jump (sic) to jump-based stability notions. We first observe that the structure of the utility functions of the agents plays a crucial role in the existence and computation of jump-stable allocations. When agents have binary utilities, jump-stable allocation do not always exist, even when the graph topology is a path and we have two agents (\Cref{thm:jump:Non-existence}) and in general it is \NPc to compute them (\Cref{thm:jump:binary:NPc}). On the contrary, when the utility functions are all monotone of the same type (all increasing or all decreasing), a jump-stable allocation can always be reached via the natural dynamics in polynomial time (\Cref{thm:jump:monotone}).

We also show several results for the more restrictive notion of~$\ell$-jump-stability according to which no agent has an incentive to jump to an empty vertex at a distance of at most~$\ell$ from her current location. In particular, we show that, when agents have binary utility functions, a~$1$-jump-stable allocation always exists on tree topologies and can be computed efficiently via a simple greedy algorithm (\Cref{thm:oneJump:binary:trees}), but does not necessarily exist when the graph contains cycles, even for identical binary, or single peaked, utilities (\Cref{thm:1-jump-counterexample}). 
In addition, we prove that it is \NPc to decide the existence of a~$2$-jump-stable allocation when agents have binary utilities (\Cref{thm:twoJump:binary:NPc}). The complexity of~$1$-jump stable allocations remains open.

We finally consider the problem of maximizing social welfare for several different definitions. Specifically, we focus on the utilitarian social welfare (USW; the total utility of the agents), the egalitarian social welfare (ESW; the minimum utility over all agents), the Nash social welfare (NSW; the geometric mean of the agents' utilities), and Pareto optimality (PO; according to this notion, there is no other allocation exists such that an agent has strictly larger utility and no other agent is worst off). We first show that it is \NPc and \Wc with respect to the number of agents to compute welfare-maximizing allocations for USW, ESW, and NSW (\Cref{thm:USW:NPc}). We then observe that even though a PO allocation always exists, it is \coNPh and \coWh under the same parameterization to verify if an allocation is indeed PO (\Cref{thm:PO:verify:numAgents:coWh}). We complement these strong impossibilities with an XP algorithm that solves all these problems and is, for the three welfare-maximization problems, asymptotically optimal under the Exponential Time Hypothesis (\Cref{thm:ESW:numAgents:XP}).

\subsection{Related Work}\label{sec:relatedWork}

As already stressed, most closely related to our model are \emph{Schelling games} \citep{schelling-journal,bilo2022topological,bilo2023continuous,ijcai2022p12,bullinger2021welfare,chauhan2018schelling,echzell2019convergence,chan2020schelling,kanellopoulos2020modified,kanellopoulos2023tolerance,Kreisel2024equilibria,deligkas2024parameterized,schelling1969models,schelling1971dynamic} and \emph{anonymous hedonic games}~\citep{BanerjeeKS2001,BogomolnaiaJ2002,Ballester2004,ConstantinescuLRSV2024,FioravantesGMS2026}.
In contrast to Schelling games, where the agents are partitioned into types and their utilities are only implicitly determined by the type composition of their neighborhood\footnote{The vast majority of papers focus on agents with homophilic preferences. However, recently, diversity- and variety-seeking preferences were also considered~\citep{NarayananOTV2025,NarayananSV2025}.}, we work directly with explicit utility functions, allowing each agent an arbitrary function of the number of occupied neighbors. In contrast to anonymous hedonic games, which likewise involve no types but partition the agents into disjoint coalitions, our agents are placed on a topology in which their neighborhoods overlap. A similar contrast applies to \emph{load balancing games}~\citep{Vocking2007,EvenDarKM2007}, where an agent's cost depends on the number (or load) of agents assigned to the same machine: the machines are disjoint bins with no topology among them, which is the key reason why stable outcomes there are guaranteed by potential-function arguments, whereas in our game they may fail to exist.

Anonymous, size-based preferences are also at the core of the \emph{group activity selection problem} (GASP)~\citep{DarmannEKLSW2018}, introduced as a natural generalization of anonymous hedonic games, in which each agent has preferences over pairs consisting of an activity and the number of its participants. Its variant on social networks, gGASP~\citep{IgarashiPE2017}, additionally requires the group engaging in an activity to form a connected subgraph of a network on the agents. In both models, however, agents are assigned to activities, and the network (when present) only constrains which coalitions are feasible; in our game, the agents are instead allocated to the vertices of the topology itself, which determines their neighborhoods and hence their utilities.

A related class of problems is the so-called \emph{refugee housing}~\citep{KnopS2023,Schierreich2023,Schierreich2024,LisowskiS2025}, where agents are likewise placed on the empty vertices of a graph topology and an agent's satisfaction depends on the number of agents in her neighborhood. The crucial difference is that these models feature a set of preallocated agents (the inhabitants) who are fixed in place, whereas in our game, every agent is free to deviate. This property makes both classes very different. Moreover, the anonymous preferences studied in refugee housing take the form of approved sets or upper bounds on the number of accepted neighbors, while we allow for much more general classes of utility functions.

A further group of games places agents on a structure and derives the utility of an agent from the agents located near her. In the \emph{seat arrangement problem}~\citep{BodlaenderHJOOZ2025,Wilczynski2023,BerriaudCW2023,CeylanCR2026}, agents are assigned to seats and care about who exactly is placed next to them, whereas in \emph{social distance games}~\citep{BranzeiL2011,KaklamanisKP2018,BalliuFMO2019,BalliuFMO2022,GanianHKRSS2023}, an agent's utility aggregates her distances to the other agents in her connected component. Finally, in \emph{topological distance games}~\citep{BullingerS2023,DeligkasEKS2024b} and \emph{distance preservation games}~\citep{AzizCLNW2025,DeligkasEGKS2026}, the utility of an agent depends on her graph distances to specific other agents. In contrast to all these distance-based models, the utility in our game is \emph{anonymous} and \emph{local}: it is determined solely by the number of occupied vertices in an agent's immediate neighborhood.

\section{The Model}\label{sec:model}
We consider a network allocation game~$\Gamma = (G, \agents, (\util_i)_{i \in \agents})$ consisting of a set~$\agents$ of~$\numAgents$ \emph{agents} and an undirected graph~$G=(V, E)$ with~$|V| \geq \numAgents$ to which we refer as the {\em topology}. By~$N_G(v) = \{u \in V \mid uv \in E\}$ we denote the \emph{neighborhood} of a vertex~$v\in V$, i.e., the set of all vertices adjacent to~$v$. The goal is to assign the agents to the vertices of the topology. Formally, an \emph{allocation} is an injective function~$\alloc\colon \agents\to V$ that assigns to each agent her host vertex; we also call it the \emph{location} of the agent. Observe that, by injectivity, every vertex hosts at most one agent; that is, the agents cannot share a vertex.
Additionally, each agent~$i\in\agents$ is associated with a \emph{utility function}~${\util_i\colon [\numAgents-1]_0\to\R_0^+}$. Intuitively, the function returns the utility agent~$i$ experiences if there are~$x\in[\numAgents-1]_0$ agents allocated to her direct neighborhood. We naturally extend the notation from numbers to allocations and set~$\util_i(\alloc) = \util_i(|\{ j \in \agents \mid \alloc(j) \in N_G(\alloc(i)) \}|)$. We say that the utility function~$\util_i$ is:

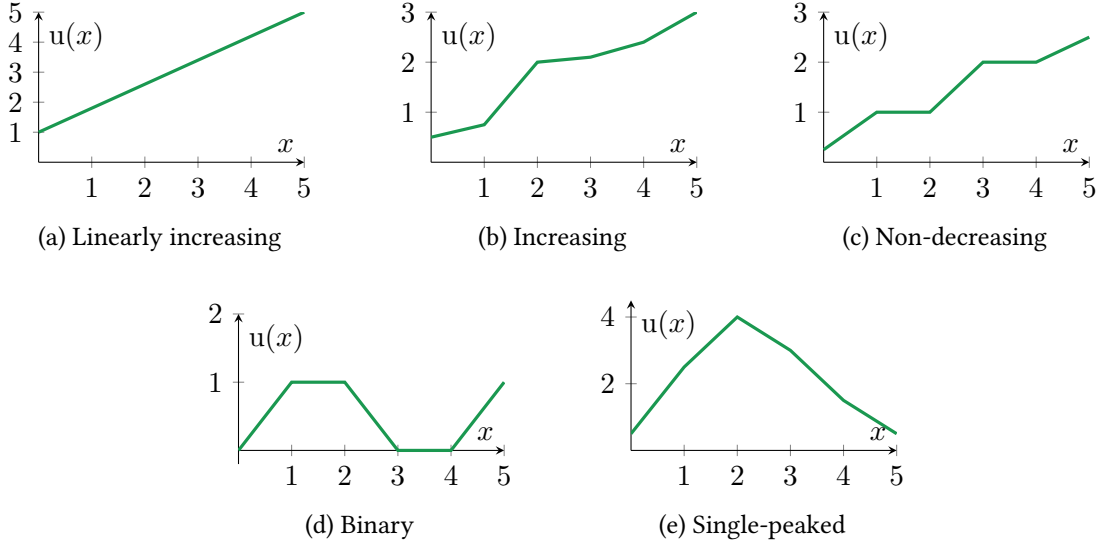
\begin{figure}[tb!]
    \centering
    \begin{subfigure}{0.32\textwidth}
    \centering
    \begin{tikzpicture}
    \begin{axis}[
        axis lines=middle,
        xmin=0, xmax=5,
        ymin=0, ymax=5,
        xtick={0,1,2,3,4,5},
        ytick={0,1,2,3,4,5},
        xlabel={$x$},
        ylabel={$\util(x)$},
        width=\textwidth,
        height=0.7\textwidth
    ]
        \addplot[cbGood, very thick] coordinates {
            (0,1) (1,1.8) (2,2.6) (3,3.4) (4,4.2) (5,5)
        };
    \end{axis}
    \end{tikzpicture}
    \caption{Linearly increasing}
    \end{subfigure}
    \begin{subfigure}{0.32\textwidth}
    \centering
    \begin{tikzpicture}
    \begin{axis}[
        axis lines=middle,
        xmin=0, xmax=5,
        ymin=0, ymax=3,
        xtick={0,1,2,3,4,5},
        xlabel={$x$},
        ylabel={$\util(x)$},
        width=\textwidth,
        height=0.7\textwidth
    ]
        \addplot[cbGood, very thick] coordinates {
            (0,0.5) (1,0.75) (2,2) (3,2.1) (4,2.4) (5,3)
        };
    \end{axis}
    \end{tikzpicture}
    \caption{Increasing}
    \end{subfigure}
    \begin{subfigure}{0.32\textwidth}
    \centering
    \begin{tikzpicture}
    \begin{axis}[
        axis lines=middle,
        xmin=0, xmax=5,
        ymin=0, ymax=3,
        xtick={0,1,2,3,4,5},
        xlabel={$x$},
        ylabel={$\util(x)$},
        width=\textwidth,
        height=0.7\textwidth
    ]
        \addplot[cbGood, very thick] coordinates {
            (0,0.25) (1,1) (2,1) (3,2) (4,2) (5,2.5)
        };
    \end{axis}
    \end{tikzpicture}
    \caption{Non-decreasing}
    \end{subfigure}
    \vspace{0.5cm}

    \begin{subfigure}{0.32\textwidth}
    \centering
    \begin{tikzpicture}
    \begin{axis}[
        axis lines=middle,
        xmin=0, xmax=5,
        ymin=-0.2, ymax=2,
        xtick={0,1,2,3,4,5},
        ytick={0,1,2},
        xlabel={$x$},
        ylabel={$\util(x)$},
        width=\textwidth,
        height=0.7\textwidth
    ]
        \addplot[cbGood, very thick] coordinates {
            (0,0) (1,1) (2,1) (3,0) (4,0) (5,1)
        };
    \end{axis}
    \end{tikzpicture}
    \caption{Binary}
    \end{subfigure}
    \begin{subfigure}{0.32\textwidth}
    \centering
    \begin{tikzpicture}
    \begin{axis}[
        axis lines=middle,
        xmin=0, xmax=5,
        ymin=0, ymax=4.5,
        xtick={0,1,2,3,4,5},
        xlabel={$x$},
        ylabel={$\util(x)$},
        width=\textwidth,
        height=0.7\textwidth
    ]
        \addplot[cbGood, very thick] coordinates {
            (0,0.5) (1,2.5) (2,4) (3,3) (4,1.5) (5,0.5)
        };
    \end{axis}
    \end{tikzpicture}
    \caption{Single-peaked}
    \end{subfigure}
        
    \caption{Illustrations of the restrictions of utility functions used throughout the paper. Additionally, we define linearly decreasing, decreasing, and non-increasing utilities, which are natural counterparts of linearly increasing, increasing, and non-decreasing utilities, respectively.}
    \label{fig:restrictions:examples}
\end{figure}

\begin{itemize}
    \item \emph{linearly increasing} if there exist~$a\in\N$ and~$b\in\N_0$ such that~$\util_i(x) = a\cdot x + b$ for every~$x\in[\numAgents-1]_0$;
    
    \item \emph{increasing} if~$\util_i(x) < \util_i(y)$ whenever~$x < y$;
    
    \item \emph{non-decreasing} if~$\util_i(x) \leq \util_i(y)$ whenever~$x<y$;

    \item \emph{linearly decreasing} if there exist~$a\in\N$ and~$b\in\N_0$ with~$b \geq a\cdot(\numAgents-1)$ such that~$\util_i(x) = b - a\cdot x$ for every~$x\in[\numAgents-1]_0$;

    \item \emph{decreasing} if~$\util_i(x) > \util_i(y)$ whenever~$x < y$;

    \item \emph{non-increasing} if~$\util_i(x) \geq \util_i(y)$ whenever~$x < y$;

    \item \emph{binary} if~$\util_i(x) \in \{0,1\}$ for every~$x\in[\numAgents-1]_0$;
    
    \item \emph{single-peaked} if there exists~$z^*\in[\numAgents-1]_0$ such that for every~$x,y\in[0,z^*]$,~$x<y$, it holds that~$\util_i(x) < \util_i(y)$, and for every~$x,y\in[z^*,\numAgents-1]$,~$x < y$, we have~$\util_i(x) > \util_i(y)$.%
    
\end{itemize}

\noindent 
See \Cref{fig:restrictions:examples} for an illustration of the restrictions of utility functions defined above and \Cref{fig:restrictions:relations} for their relationships. Additionally, we say that the utility function is \emph{normalized} if~$\util_i(0) = 0$, \emph{monotone} if it is non-decreasing or non-increasing, and \emph{participation appreciative} if~$\util_i(0) < \util_i(x)$ for every~$x\in[\numAgents-1]$. Observe that every increasing function (and therefore also every linearly increasing function) is also participation appreciative.

Binary utilities correspond to the well-studied dichotomous (approval) preferences~\citep{Peters2016}: an agent simply approves a set of acceptable neighborhood sizes. Single-peaked utilities capture an ideal amount of social interaction, with satisfaction decreasing in either direction away from the peak; single-peakedness is a classic domain restriction that was recently also studied in Schelling games~\citep{friedrich2023single}.

\begin{figure}[bt!]
    \centering
    \vspace{0.75cm}
    \begin{tikzpicture}[scale=0.85]
    
        \node (l) at (-6,0) {linearly increasing};
        \node (i) at (-6,1.25) {increasing};
        \node (m) at (-6,3.5) {non-decreasing};

        \node (b) at (6, 3.5) {binary};

        \node (spe) at (-2,3) {single-peaked};

        \node (ld) at (2,0) {linearly decreasing};
        \node (sd) at (2,1.25) {decreasing};
        \node (d) at (2,3.5) {non-increasing};

        \node (g) at (0, 5) {unrestricted};
        
        \draw[->] (l) -- (i);
        \draw[->] (i) -- (m);

        \draw[->] (ld) -- (sd);
        \draw[->] (sd) -- (d);
        \draw[->] (d) -- (g);
        \draw[->] (sd) -- (spe);

        \draw[<-] (g) -- (m);
        \draw[<-] (g.south) to (spe);
        \draw[<-] (g) -- (b);
        \draw[<-] (spe) -- (i);
    \end{tikzpicture}
    \caption{Relationships between our utility restrictions. An arrow from~$A$ to~$B$ means that whenever a utility function satisfies restriction~$A$, it also satisfies~$B$.}
    \label{fig:restrictions:relations}
\end{figure}
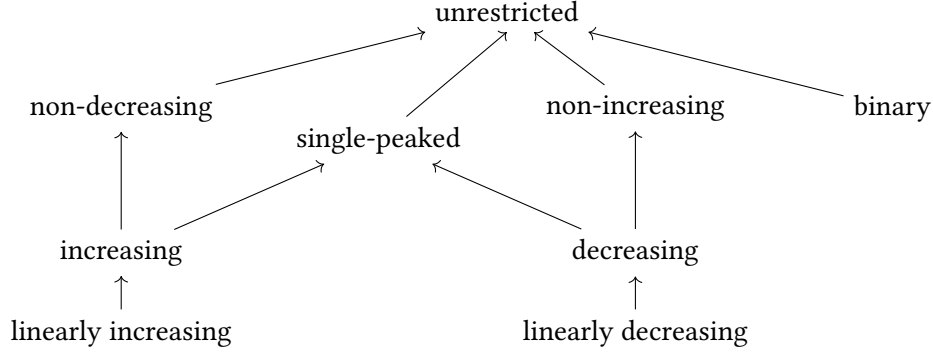

\subsection{Stability Notions}
We are interested in allocations that are \emph{stable}.
The first notion of stability we study is that of \emph{envy-freeness}, which is extensively studied in the fair division literature~\citep{AmanatidisABFLMVW2023,BramsT1996,Foley1967}, but was recently also considered in settings similar to ours; see, e.g.,~\citep{DeligkasEGKS2026,DeligkasEIKS2025,LiMN02023}. For~$\alpha \in [0,1]$, we say that agent~$i$ \emph{$\alpha$-envies}, agent~$j\neq i$ if~$\util_{i}(\alloc) < \alpha\cdot \util_i(\alloc^{i \leftrightarrow j})$, where~$\alloc^{i\leftrightarrow j}$ is the allocation obtained from~$\alloc$ by swapping only the locations of~$i$ and~$j$. An allocation without~$\alpha$-envious agents is called \emph{$\alpha$-envy-free}. If~$\alpha=1$, we will say that agent~$i$ \emph{envies} agent~$j$ instead of~$i$~$1$-envies~$j$, and an allocation is \emph{envy-free} instead of~$1$-envy-free, respectively.

One can argue that the requirement of envy-freeness, i.e., that no agent prefers another's location to her own, is too demanding, as it ignores whether the swap would be agreeable to the other party. Therefore, we also study a weaker notion of stability, in which both agents must agree to the potential swap of locations. Formally, a pair of agents~$i,j$ admits a \emph{swap deviation} if~$\util_i(\alloc) < \util_i(\alloc^{i\leftrightarrow j})$ and~$\util_j(\alloc) < \util_j(\alloc^{i\leftrightarrow j})$. An allocation in which no pair of agents admits a swap deviation is called \emph{swap stable}. A \emph{weak swap deviation} is then defined analogously, but we allow two agents to swap their locations when neither of them is worse off and at least one of them is strictly better off. Finally, an agent~$i$ can perform a~$\beta$-hurting swap deviation,~$\beta \in [0,1]$, with an agent~$j$ if~$\util_i(\alloc) < \util_i(\alloc^{i\leftrightarrow j})$ and~$\util_j(\alloc^{i\leftrightarrow j}) \geq (1-\beta)\cdot\util_j(\alloc)$, i.e., if~$i$ strictly improves her utility and the utility of~$j$ decreases by at most a factor of~$\beta$. Observe that if~$\beta = 0$, then~$\beta$-hurting swap deviations are equivalent to weak swap deviations, and when~$\beta = 1$, then a~$\beta$-hurting swap deviation is equivalent to~$i$ envying~$j$.

Finally, we say that an agent~$i$ admits a \emph{jump deviation} to an empty vertex~$v$ if~$\util_i(\alloc) < \util_i(\alloc^{i\mapsto v})$, where~$\alloc^{i\mapsto v}$ is the allocation obtained from~$\alloc$ by changing the location of~$i$ to~$v$. An allocation with no jump deviating agent is called \emph{jump stable}. We additionally restrict the space of possible jump deviations by introducing a parameter~$\ell$. Formally, an \emph{$\ell$-jump deviation} is a jump deviation of agent~$i$ to a vertex~$v$ such that~$\dist_G(\alloc(i),v) \leq \ell$, where~$\dist_G(u,w)$ is the length of the shortest path between~$u$ and~$w$ in~$G$; we set~$\dist_G(u,w) = \infty$ if no such path exists. Then, an allocation is \emph{$\ell$-jump stable} if no~$\ell$-jump deviation is possible.

\subsection{Efficiency Notions}
We also study several notions of economic efficiency, which are based on aggregating the utilities of the agents. 
In particular, we say that an allocation~$\alloc$ 
\begin{itemize}
    \item \emph{maximizes the utilitarian social welfare} (USW) if, for every allocation~$\alloc'$, 
    \[
        \USW(\alloc) \coloneq \sum_{i\in\agents} \util_i(\alloc) \geq \USW(\alloc')\,.
    \]

    \item \emph{maximizes the egalitarian social welfare} (ESW) if, for every allocation~$\alloc'$,
    \[
        \min_{i\in\agents} \util_i(\alloc) \geq \min_{i\in\agents}\util_i(\alloc')\,.
    \]

    \item \emph{maximizes the Nash social welfare} (NSW) if, for every allocation~$\alloc'$,
    \[
        \NSW(\alloc) \coloneq \prod_{i\in\agents} \util_i(\alloc) \geq \NSW(\alloc')\,.
    \]

    \item is \emph{Pareto optimal} (PO) if there is no allocation~$\alloc'$ such that
    \[
        \forall i \in \agents\colon \util_i(\alloc') \geq \util_i(\alloc) \text{ \ \ and \ \ } \exists i \in \agents\colon \util_i(\alloc') > \util_i(\alloc).
    \]
    \end{itemize}

\section{Swap-Based Deviations}\label{sec:swaps}

In this section, we study swap-based stability solutions, and we prove a dichotomy that depends on the solution concept; see Figure~\ref{fig:overview:swapResults}.

\begin{figure}[bt!]
    \centering
    \begin{tikzpicture}
        \node (EF) at (0,0) {envy-free};
        \node (aEF) at (-2,-2) {$\alpha$-envy-free};
        \node (zEF) at (-2,-4) {$0$-envy-free};

        \node (hsw) at (2,-1.25) {$\beta$-hurting swap};
        \node (wsw) at (2,-2.75) {weakly swap stable};
        \node (sw) at (2,-4) {swap stable};

        \draw[->] (EF) -- (aEF);
        \draw[->] (EF) -- (hsw);
        \draw[->] (aEF) -- (zEF);
        \draw[->] (hsw) -- (wsw);
        \draw[->] (wsw) -- (sw);

        \begin{pgfonlayer}{background}
            \draw[cbGood,very thick,fill=cbGood,fill opacity=0.15,rounded corners] ($(wsw.north east)+(0.25,0)$) rectangle ($(zEF.south west)+(-0.25,0)$); 
            \draw[cbHard,very thick,fill=cbHard,fill opacity=0.45,rounded corners] ($(aEF.south west)+(-0.2,0)$) rectangle ($(hsw.north east)+(0.5,1.25)$);
        \end{pgfonlayer}
    \end{tikzpicture}
    \caption{An overview of our results for swap-based deviations. A red background indicates non-existence and \NPhness, even if the utilities are binary, linearly increasing, or linearly decreasing (for~$\beta$-hurting swap stability, increasing or decreasing utilities only); a green background indicates existence and polynomial-time solvability (even for unrestricted utilities).}
    \label{fig:overview:swapResults}
\end{figure}
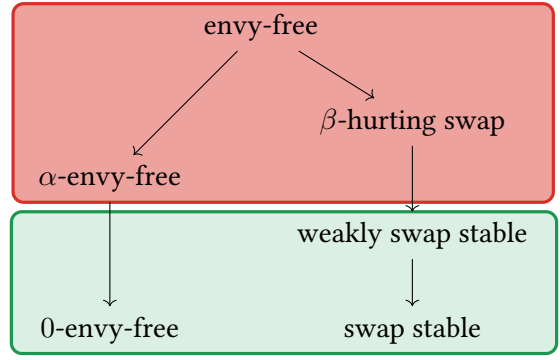

We begin by proving that swap stable allocations always exist no matter the utility functions of the agents, and can be efficiently found via natural dynamics.

\begin{theorem}
\label{thm:swap_stable}
    Every network allocation game~$\Gamma$ admits at least one (weakly) swap stable allocation. Moreover, the allocation can be found in polynomial time. 
\end{theorem}
\begin{proof}
    Starting from an arbitrary allocation, we repeatedly allow a pair of agents~$i,j$ to swap locations if their utility (or at least one of them in the case of weak swap stability) increases. Clearly, if there is no (weakly) profitable swap, then the allocation is already swap stable. Moreover, observe that the utilitarian social welfare (the total utility of the agents) increases with each swap:~$i$ and~$j$ increase their utility, and the utility of any agent different than~$i$ and~$j$ remains the same (since it only depends on the number of her neighbors, which does not change once~$i$ and~$j$ swap their locations). Consequently, the dynamics converges to a swap stable allocation.
    Furthermore, observe that no agent's utility ever decreases along the dynamics, while each swap strictly increases the utility of at least one agent. Since the utility function of each agent takes at most~$\numAgents$ distinct values, each agent's utility can strictly increase at most~$\numAgents - 1$ times, and hence the dynamics converges to a (weakly) swap stable allocation after at most~$\numAgents(\numAgents - 1) \in \Oh{\numAgents^2}$ swaps. 
\end{proof}

We note that our proof implies a stronger result. In every anonymous network allocation game, applying (weak) swap dynamics in any initial allocation converges to a (weak) swap stable allocation after at most~$\numAgents^2$ steps.

Next, we switch our attention to envy-free allocations. We begin by observing that when there are only two agents, any allocation is trivially envy-free; this follows from the fact that the two locations are symmetric, i.e., they either both have zero neighbors, or one.

\begin{observation}
    If there are two agents, any allocation is envy-free.
\end{observation}

On the other hand, though, as we show next, when there are three or more agents, envy-free allocations fail to exist even in the approximate sense, for any non-trivial approximation.

\begin{proposition}
    \label{thm:alpha-envy-free:Non-existence}
    For any~$\alpha\in(0,1]$, an~$\alpha$-envy-free allocation is not guaranteed to exist, even with three agents and identical utilities that are either binary, linearly increasing, or linearly decreasing.
\end{proposition}
\begin{proof}
    For each case, we exhibit a topology and a utility function under which no allocation is~$\alpha$-envy-free for any~$\alpha\in(0,1]$.

    \proofparagraph{Binary utilities.}
    Let the topology be a star with center~$c$ and two leaves~$\ell_1, \ell_2$, and let~$\agents = \{a_1, a_2, a_3\}$ with identical binary utility functions such that~$\util(0) = \util(1) = 0$,~$\util(2) = 1$. By the symmetry of the leaves, every allocation places one agent at~$c$ and the other two at the leaves. Assume without loss of generality that~$\alloc(a_3) = c$,~$\alloc(a_1) = \ell_1$,~$\alloc(a_2) = \ell_2$. Then,~$\util_{a_1}(\alloc) = \util(1) = 0$ and~$\util_{a_1}(\alloc^{a_1 \leftrightarrow a_3}) = \util(2) = 1$, so~$\util_{a_1}(\alloc) = 0 < \alpha \cdot 1$ for every~$\alpha \in (0, 1]$. Hence,~$a_1$~$\alpha$-envies~$a_3$ for any~$\alpha\in(0,1]$.

    \proofparagraph{Linearly decreasing utilities.}
    We use the same star topology as in the binary utilities case, but we replace the utilities with~$\util(x) = 2 - x$. In any allocation, one agent is at~$c$, and the other two are at the leaves; assume~$\alloc(a_3) = c$. Then~$\util_{a_3}(\alloc) = \util(2) = 0$ and~$\util_{a_3}(\alloc^{a_3 \leftrightarrow a_1}) = \util(1) = 1$, so~$a_3$~$\alpha$-envies~$a_1$ for every~$\alpha \in (0, 1]$.

    \proofparagraph{Linearly increasing utilities.}
    Let the topology be the disjoint union of two edges~$\{u_1,u_2\}$ and~$\{w_1, w_2\}$. We use the utility function~$\util(x) = x$, which is clearly linearly increasing. In any allocation of three agents on four vertices, exactly one edge contains two agents and the other contains one. Assume without loss of generality~$\alloc(a_1) = u_1$,~$\alloc(a_2) = u_2$,~$\alloc(a_3) = w_1$. Then~$\util_{a_3}(\alloc) = \util(0) = 0$. Since~$\util_{a_3}(\alloc^{a_1 \leftrightarrow a_3}) = \util_{a_3}(\alloc^{a_2 \leftrightarrow a_3}) = \util(1) = 1$,~$a_3$~$\alpha$-envies both~$a_1$ and~$a_2$ for every~$\alpha \in (0, 1]$.
\end{proof}

Trivially, if the factor~$\alpha$ from the previous proposition is equal to zero, which is the only missing case, then by the non-negativity of utilities, we obtain that any allocation is~$0$-envy-free.

\begin{observation}
    Any allocation is~$0$-envy-free.
\end{observation}

A similar non-existence result holds for~$\beta$-hurting swap stability, for every non-trivial value of the parameter~$\beta$.

\begin{proposition}\label{thm:beta-hurting:Non-existence}
    For every~$\beta\in(0,1)$, a~$\beta$-hurting swap stable allocation is not guaranteed to exist, even with three agents and identical utilities that are either increasing or decreasing.
\end{proposition}
\begin{proof}
    Fix~$\beta\in(0,1)$. For increasing utilities, let the topology be a star with center~$c$ and two leaves~$\ell_1, \ell_2$, and let three agents have the identical increasing utility function~$\util(x) = 2(1-\beta) + \beta x$. In every allocation, some agent~$i$ occupies the center and has both other agents as neighbors, so~$\util_i(\alloc) = \util(2) = 2$, while each of the other two agents has exactly one occupied neighbor and utility~$\util(1) = 2 - \beta$. Let~$j$ be an agent on a leaf. By swapping with~$i$, agent~$j$ obtains utility~$\util(2) = 2 > 2 - \beta = \util_j(\alloc)$, so she strictly improves, while agent~$i$ ends up with a single occupied neighbor and utility~$\util(1) = 2 - \beta \geq 2(1-\beta) = (1-\beta)\cdot\util_i(\alloc)$. Hence,~$j$ can perform a~$\beta$-hurting swap deviation with~$i$, and no allocation is~$\beta$-hurting swap stable.

    For decreasing utilities, we take the complement of the topology --- a single edge~$\ell_1\ell_2$ together with an isolated vertex~$c$ --- and the utility function~$\util'(x) = \util(2-x) = 2 - \beta x$. In every allocation, the agent~$i$ on~$c$ has no occupied neighbors and utility~$\util'(0) = 2$, while the other two agents have one occupied neighbor each and utility~$\util'(1) = 2 - \beta$. As before, either of the latter two agents strictly improves by swapping with~$i$, whose utility after the swap is~$\util'(1) = 2 - \beta \geq (1-\beta)\cdot\util'(0)$; hence, no allocation is~$\beta$-hurting swap stable.
\end{proof}

Next, we show that the dichotomy between (weak) swap-stability and (relaxations) of envy-freeness extends beyond the landscape of existence and also includes computational aspects of these notions. Specifically, in sharp contrast to weakly swap stable solutions, deciding whether an instance admits an~$\alpha$-envy-free allocation for some~$\alpha\in(0,1]$, or a~$\beta$-hurting swap stable allocation for some~$\beta\in(0,1)$ is \NPc, even for the most restricted models of preferences.

\begin{theorem}
    \label{thm:alpha-envy-free:Hardness}
    It is \NPc to decide whether an instance~$\mathcal{I}$ admits at least one envy-free allocation, even if the utilities are (i) binary, (ii) linearly increasing, or (iii) linearly decreasing.
    The same reduction also shows \NPcness for~$\alpha$-envy-freeness,~$\alpha\in(0,1]$, and~$\beta$-hurting swap stability,~$\beta\in(0,1)$, but only if the preferences are increasing, decreasing, or, in the case of~$\alpha$-envy-freeness, binary.
\end{theorem}
\begin{proof}
    All problem variants are contained in \NP: given an allocation~$\alloc$ as a certificate, one can verify in polynomial time that no agent ($\alpha$-)envies another agent, or that no pair of agents admits a~$\beta$-hurting swap deviation, by inspecting all ordered pairs of agents. In the remainder of the proof, we show hardness.
    We will prove the theorem via a reduction from the \NPc \probName{3SAT} problem~\citep{Karp1972}. We will explain the reduction for increasing functions; all results for decreasing functions can be obtained analogously by taking the complement of the graph.  
    An instance of \probName{3SAT} consists of a Boolean formula in CNF with~$n$ variables and~$m$ clauses, where each clause contains three literals. 
    We will construct an instance~$\Gamma = (G,N,(u_i)_{i\in N})$ of our problem as follows.
    For each clause~$C$ and each literal~$x$ of the clause, we add to~$G$ a vertex~$v_{C,x}$ and we add the edges~$v_{C_1,x}v_{C_2,y}$, if~$y \neq \bar{x}$ and~$C_1\neq C_2$.
    Moreover, we add to~$G$ a set~$W$ of~$2m+1$ vertices such that for all~$w\in W$ and all~$v\in V(G)\setminus\{w\}$ the edge~$wv$ is in~$E(G)$.    
    The set~$N$ contains~$3m+1$ agents with identical utility functions:
    \begin{itemize}
        \item for binary utility function (this is the same for ~$\alpha$-envy-freeness), each agent has utility 1, if she has~$3m$ occupied neighboring vertices;
        \item for linearly increasing utilities, we set~$u_i(x) = x$; that is, the utility of an agent equals the number of occupied neighboring vertices;
        \item for~$\alpha$-envy-freeness and increasing utility, we set~$u_i(3m) = \frac{3m}{\alpha}$ and~$u_i(x) = x$ otherwise;
        \item for~$\beta$-hurting swap stability and increasing utility, we set~$u_i(x) = 3m(1-\beta) + x\beta$.
    \end{itemize}
    We claim that there is a solution to our problem if and only if the \probName{3SAT} instance is satisfiable. Firstly, assume that the Boolean is satisfiable. To find an envy-free allocation, we place~$2m+1$ agents on~$W$. Without loss of generality, let the remaining agents be~$1, 2, \ldots, m$ and let the clauses be~$C_1, C_2, \ldots, C_m$. Let~$x_i$ be an arbitrary satisfied literal in~$C_i$. We place the agent~$i$ on a vertex~$v_{C_i, x_i}$. Note that a satisfying assignment cannot have the literal and its negation satisfied at the same time. So the agents are allocated on a clique with~$3m+1$ vertices. Hence, the allocation is envy-free and therefore also~$\alpha$-envy-free and~$\beta$-hurting swap stable.

    In the other direction, notice that we have~$3m+1$ agents but only~$3m$ vertices associated with a literal of the \probName{3SAT} instance. Hence, at least one agent is allocated to~$W$ and so has all~$3m$ other agents in the neighborhood. 
    Let~$i$ be an arbitrary agent allocated on a vertex in~$W$. Let~$j$ be an arbitrary other agent. We first show that all agents in~$N\setminus \{j\}$ have to be allocated to the neighborhood of~$j$. 
    In the case of binary utility, this is trivial as~$u_j(3m) = 1$ and~$u_j(x) = 0$ otherwise, so~$j$ would ($\alpha$-)envy~$i$ if they did not have all agents in their neighborhood. Similarly, for the linearly increasing envy-freeness case,~$j$ would envy~$i$, if they do not get all agents in their neighborhood. 
    For the case of~$\alpha$-envy-free and increasing utility, notice that by swapping with~$i$, agent~$j$ would obtain utility~$u_j(3m) = \frac{3m}{\alpha}$, so~$j$~$\alpha$-envies~$i$ whenever~$u_j(\alloc) < \alpha\cdot\frac{3m}{\alpha} = 3m$, which is the case unless~$j$ has all~$3m$ other agents in her neighborhood. Finally, for~$\beta$-hurting swap stability, notice that~$u_i(3m) = 3m$ and~$u_i(x) \geq 3m(1-\beta)$ for every~$x$, with equality if and only if~$x = 0$. Hence, if some agent~$j$ had fewer than~$3m$ occupied neighboring vertices, then~$u_j(\alloc) < 3m$, so~$j$ would strictly improve by swapping with~$i$, while the utility of~$i$ after the swap would be~$u_i(x) \geq 3m(1-\beta) = (1-\beta)\cdot u_i(\alloc)$ regardless of the number~$x$ of occupied vertices neighboring her new location. That is,~$j$ could perform a~$\beta$-hurting swap deviation with~$i$, and the allocation would not be stable. 
    It follows that the agents have to be allocated to a clique. Since for each clause~$C=\{x,y,z\}$, the vertices~$v_{C,x}, v_{C,y}, v_{C,z}$ are independent, there has to be exactly one agent allocated to each clause, and the remaining~$2m+1$ agents are allocated to~$W$. Moreover, a pair of agents cannot be allocated to a literal and its negation in two different clauses, as the associated vertices would not have an edge between them. Therefore, for each clause, we can assign the literal with the agent true to obtain a satisfying assignment for the original \probName{3SAT} instance.

    To obtain the results for decreasing functions, we can reduce directly from the analogous case for an increasing function by taking the complement of the graph~$G$ and letting~$u_i'(x) = u_i(3m-x)$. The correctness follows straightforwardly from the fact that if an agent has~$x$ agents in the neighborhood in~$G$, it has~$3m-x$ agents in the neighborhood in the complement of~$G$.
\end{proof}

The following remark explains why binary utilities are excluded for~$\beta$-hurting swap stability in \Cref{thm:alpha-envy-free:Hardness}: in this case, the notion collapses to weak swap stability, and no hardness result is possible.

\begin{remark}
    For binary utilities and any~$\beta \in (0,1)$, a~$\beta$-hurting swap stable allocation always exists and can be found in polynomial time. Indeed, if~$\util_j(\alloc) = 0$, then the condition~$\util_j(\alloc^{i\leftrightarrow j}) \geq (1-\beta)\cdot\util_j(\alloc)$ is trivially satisfied, and if~$\util_j(\alloc) = 1$, then it is satisfied if and only if~$\util_j(\alloc^{i\leftrightarrow j}) = 1$, since~$1-\beta \in (0,1)$ and utilities take values in~$\{0,1\}$. In both cases, the condition is equivalent to~$\util_j(\alloc^{i\leftrightarrow j}) \geq \util_j(\alloc)$. Hence, under binary utilities,~$\beta$-hurting swap deviations coincide with weak swap deviations, and the claim follows from \Cref{thm:swap_stable}.
\end{remark}

We conclude the section with one more positive result. Specifically, we show that the computational hardness of the problem crucially stems from a strategic decision of how to allocate a large number of agents. It turns out that, when the number of agents is small, deciding the existence becomes tractable.

\begin{theorem} \label{thm:alpha-envy-free:FPT}
    Deciding whether an~$\alpha$-envy-free,~$\alpha\in(0,1]$, or a~$\beta$-hurting swap stable,~$\beta\in[0,1)$, allocation exists is fixed-parameter tractable when parameterized by the number of agents~$\numAgents$.
\end{theorem}
\begin{proof}
    Let~$R(\numAgents, \numAgents)$ denote the Ramsey number: the smallest~$\nu$ such that every graph on~$\nu$ vertices contains either a clique of size~$\numAgents$ or an independent set of size~$\numAgents$. By the Erdős-Szekeres bound~\citep{ErdosS1935},~$R(\numAgents, \numAgents) \leq 4^\numAgents$, and a homogeneous set of size~$\numAgents$ can be found in time polynomial in~$|V(G)|$, whenever~$|V(G)| \geq R(\numAgents, \numAgents)$. We distinguish two cases based on~$|V(G)|$.

    First, let~$|V(G)| \geq R(\numAgents, \numAgents)$. Then, by the previous arguments, we can find either a clique~$K$ or an independent set~$I$ of size~$\numAgents$ in~$G$. Let~$\alloc$ be any bijective allocation of the agents onto the vertices of~$K$ or~$I$. In either case, every agent has the same number of allocated neighbors as every other agent in~$\alloc$ and in every allocation~$\alloc^{i \leftrightarrow j}$ obtained by swapping a pair of agents. Hence,~$\util_i(\alloc) = \util_i(\alloc^{i \leftrightarrow j})$ for every~$i,j \in \agents$, so the allocation is~$1$-envy-free and~$0$-hurting swap stable.

    Now, let~$|V(G)| < R(\numAgents, \numAgents)$. Since~$R(\numAgents,\numAgents) < 4^\numAgents$, the number of allocations is at most~$|V(G)|^\numAgents < 4^{\numAgents^2}$. We enumerate all allocations and, for each of them, check whether it is~$\alpha$-envy-free (respectively~$\beta$-hurting swap stable) by examining all~$\Oh{\numAgents^2}$ ordered pairs of agents in polynomial time. The total running time is~$4^{\numAgents^2} \cdot |V(G)|^\Oh{1}$, which is \FPT with respect to the number of agents~$\numAgents$.
\end{proof}

\section{Jump-Based Deviations}\label{sec:jumps}

We now switch to jump-stable allocations; recall that, in such an allocation, no agent has incentive to deviate to an empty vertex. Unfortunately, in contrast to swap-stable allocations, jump-stable allocations do not necessarily always exist. 

\begin{figure}[bt!]
    \centering
    \begin{tikzpicture}
        \node (jump) at (0,-1.25) {jump stable};
        \node (lJump) at (0,-2.5) {$\ell$-local jump stable};
        \node (twoJump) at (0,-3.75) {$2$-local jump stable};
        \node (oneJump) at (0,-5) {$1$-local jump stable};
        \draw[->] (jump) -- (lJump);
        \draw[->] (lJump) -- (twoJump);
        \draw[->] (twoJump) -- (oneJump);
        \begin{pgfonlayer}{background}
            \node[text width = 2.25cm,align=center] (mon) at (-4.25,-1.25) {non-increasing non-decreasing};
            \draw[->] (mon.south) to[in=180,out=270] (-3,-4);

            \node[text width = 2.25cm,align=center] (bsp) at (4.25,-1.25) {binary \mbox{single-peaked}};
            \draw[->] (bsp.south) to[in=0,out=270] (3,-2.5);
            \draw[->] (bsp.south) to[in=0,out=270] (3,-5);

            \draw[cbGood, very thick, fill=cbGood, fill opacity=0.15, rounded corners]
                  ($(jump.west)+(-2,0.75)$) rectangle ($(oneJump.south east)+(1,-0.25)$);

            \draw[cbHard, very thick, densely dashed, fill=cbHard, fill opacity=0.45, rounded corners]
                  ($(jump.north west)+(-1.5,0.25)$) rectangle ($(twoJump.south east)+(1.5,-0.25)$);

            \draw[cbOpen, very thick, dotted, fill=cbOpen, fill opacity=0.30, rounded corners]
                  ($(oneJump.north west)-(1,-0.15)$) rectangle ($(oneJump.south east)+(1.5,-0.15)$);
        \end{pgfonlayer}
    \end{tikzpicture}
    \caption{An overview of our solution concepts based on the jumps of single agents,
    together with the corresponding existence and complexity results. The three shaded,
    overlapping regions indicate: existence guaranteed and polynomial-time solvable
    (green, solid outline); non-existence and \NPc{} (red, dashed outline); and
    non-existence with complexity status open (yellow, dotted outline). For single-peaked preferences, only non-existence is known.}
    \label{fig:swap:results}
\end{figure}
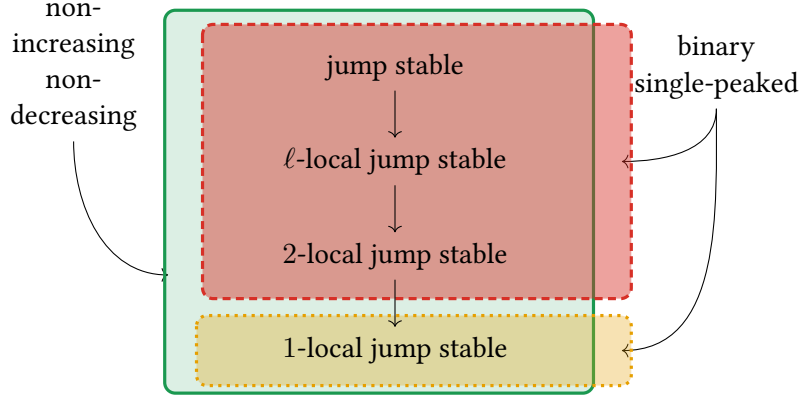

\begin{proposition}\label{thm:jump:Non-existence}
    Jump-stable allocations may fail to exist, even if there are two agents, the topology is a simple path, and the utilities are single-peaked or binary.
\end{proposition}
\begin{proof}
    Consider a game with a line graph~$G$ consisting of~$5$ vertices, and two agents. The utility function of agent~$1$ is~$u_1(0)=1$ and~$u_1(1)=0$, whereas the utility function of agent~$2$ is~$u_2(0)=0$ and~$u_2(1)=1$. To maximize her utility, agent~$2$ aims to be a neighbor of agent~$1$. In contrast, agent~$1$ aims to have no neighbors. Hence, in any possible allocation, at least one agent has an incentive to jump (either agent~$2$ jumps to be a neighbor of agent~$1$, or agent~$1$ jumps to not be a neighbor of agent~$2$).
\end{proof}

On the positive side, when the utility functions are monotone, a jump-stable allocation exists and can be computed efficiently via a natural dynamics.

\begin{theorem}\label{thm:jump:monotone}
    If the utility functions are all monotone of the same type (increasing or decreasing), a jump stable allocation is guaranteed to exist and can be found in polynomial time.
\end{theorem}
\begin{proof}
    For any allocation~$\alloc$, let~$O(\alloc) = \{z \in V \mid \exists j \in \agents\colon \alloc(j)=z \}$ be the set of non-empty vertices. In addition, for any vertex~$z \in V$, let~$n_z(\alloc) = |\{j \in \agents \mid \alloc(j) \in N_G(z)\}|$ be the number of agents that occupy vertices that are neighboring~$z$. Hence, if an agent~$i$ occupies~$v$ in~$\alloc$, she obtains utility~$u_i(\alloc) = u_i(n_v(\alloc))$. 
    We claim that the function 
    \begin{equation*}
        \Phi(\alloc) = \sum_{z \in O(\alloc)} n_z(\alloc),
    \end{equation*}
    or its negation~$-\Phi$, is an ordinal potential function. Intuitively,~$\Phi(\alloc)$ equals twice the number of edges of~$G$ with both endpoints occupied: an improving jump of an agent with a non-decreasing utility function strictly increases the number of agents neighboring her location and, by symmetry, also this number of edges, while under non-increasing utilities the opposite holds.
    Consider two allocations~$\alloc$ and~$\alloc'$ that only differ in the vertex that some agent~$i$ occupies; in particular,~$i$ occupies~$v$ in~$\alloc$ and~$v' \neq v$ in~$\alloc'$, while all other agents occupy the same vertices in the two allocations. 
    Observe that~$v'$ is empty in~$\alloc$, and~$v$ is empty in~$\alloc'$. 
    Let~$O$ be the set of non-empty vertices different than~$v$ and~$v'$ in the two allocations; clearly, these vertices are the same among the two allocations, that is,~$O = O(\alloc) \setminus \{v\} = O(\alloc') \setminus \{v'\}$.
    For an arbitrary vertex~$z \in O$, we have:
    \begin{itemize}
        \item If~$z \not\in N_G(v) \cup N_G(v')$ then~$i$ is not a neighbor of~$z$ in both allocations, and thus~$n_z(\alloc) = n_z(\alloc')$;
        \item If~$z \in N_G(v) \cap N_G(v')$, then~$i$ is a neighbor of~$z$ in both allocations, and thus~$n_z(\alloc) = n_z(\alloc')$;
        \item If~$z \in N_G(v) \setminus N_G(v')$, then~$i$ is a neighbor of~$z$ in~$\alloc$ but not in~$\alloc'$, and thus~$n_z(\alloc) = n_z(\alloc')+1$;
        \item If~$z \in N_G(v') \setminus N_G(v)$, then~$i$ is a neighbor of~$z$ in~$\alloc'$ but not in~$\alloc$, and thus~$n_z(\alloc) = n_z(\alloc')-1$.
    \end{itemize}
    Therefore, 
    \begin{align*}
        \Phi(\alloc) - \Phi(\alloc') 
        &= n_v(\alloc) - n_{v'}(\alloc') + \sum_{z \in O} \bigg ( n_z(\alloc) - n_z(\alloc') \bigg) \\
        &= n_v(\alloc) - n_{v'}(\alloc') + |O \cap (N_G(v) \setminus N_G(v'))| - |O \cap (N_G(v') \setminus N_G(v))| \\
        &= n_v(\alloc) - n_{v'}(\alloc') + n_v(\alloc) - n_{v'}(\alloc') \\
        &= 2 \bigg( n_v(\alloc) - n_{v'}(\alloc') \bigg).
    \end{align*}
    If the utility function is non-decreasing,~$\util_i(\alloc) > \util_i(\alloc')$ implies that~$n_v(\alloc) > n_{v'}(\alloc')$, and thus~$\Phi$ is a potential function. If the utility function is non-increasing,~$\util_i(\alloc) > \util_i(\alloc')$ implies that~$n_v(\alloc) < n_{v'}(\alloc')$, and thus~$-\Phi$ is a potential function.

    Finally, we argue that the dynamics converges in polynomial time. The potential~$\Phi$ is integral and satisfies~$0 \leq \Phi(\alloc) \leq \numAgents(\numAgents - 1)$ for every allocation~$\alloc$, as each of the at most~$\numAgents$ occupied vertices has at most~$\numAgents - 1$ occupied neighboring vertices. Moreover, by the computation above, every improving jump changes~$\Phi$ by~$2\left(n_{v'}(\alloc') - n_v(\alloc)\right) \neq 0$, i.e., by at least~$2$ in the direction of the respective potential. Hence, the dynamics terminates after at most~$\numAgents(\numAgents - 1)/2$ improving jumps, and in each step, an improving jump (or a certificate that none exists) can be found in polynomial time by enumerating all pairs of an agent and an empty vertex. This completes the proof.
\end{proof}

We next show that it is hard to decide if a game admits a jump-stable allocation. 
The reduction is inspired by the reduction of \citet[Theorem 9]{Peters2016} showing that it is \NPh to decide whether a given \emph{Anonymous Hedonic Game with Dichotomous Preferences} admits a Nash stable solution.

\begin{theorem}
    \label{thm:jump:binary:NPc}
    It is \NPc to decide whether a network allocation game~$\Gamma$ admits a jump stable allocation, even if all utility functions are binary.
\end{theorem}
\begin{proof}
    To show \NPhness, we reduce from the \probName{Restricted Exact Cover by~$3$-Sets} (RX3C) problem, which is known to be \NPc~\citep{Gonzalez1985}. In this problem, we are given a universe~$\mathcal{U}=\{e_1,\ldots,e_{3p}\}$ and a family~$\mathcal{S} = \{S_1,\ldots,S_{3p}\}$ of~$3$-element subsets of~$\mathcal{U}$ such that every~$e_i$ appears in exactly~$3$ sets of~$\mathcal{S}$. The question is whether there exists a set~$X\subseteq \mathcal{S}$ of size~$p$ such that~$\bigcup_{S\in X} S = \mathcal{U}$.

    Given an instance~$\mathcal{I}=(\mathcal{U},\mathcal{S})$ of RX3C, we construct an equivalent game~$\Gamma$ as follows. The high-level idea is that the agents will form clusters in the cliques of the topology, where each cluster represents either a set~$S_j \in \mathcal{S}$ chosen for the cover or an auxiliary triplet of leftover set agents; the special agent at the apex acts as a global ``anchor'', and the guard agent prevents any agent from being alone in a clique. Moreover, every agent except the special agent is also happy when fully isolated; this ensures that in a stable allocation, no agent other than the special agent can occupy the apex, as anybody else would prefer to move to an empty clique.

    For every element~$e_i\in\mathcal{U}$, we create one \emph{element agent}~$a_i$, which represents the requirement that~$e_i$ must be covered. For every set~$S_j\in\mathcal{S}$, we create~$12j-3$ \emph{set agents}~$s_j^1,\ldots,s_j^{12j-3}$; together with the three element agents corresponding to the elements of~$S_j$, these~$12j$ agents will form the ``cluster'' of~$S_j$ when~$S_j$ is selected for the cover. The factor~$12j$ varies with~$j$ to give each set a distinct preferred cluster size, ensuring that an agent in a~$12j$-cluster can only be associated with the unique set~$S_j$. To finalize the set of agents, we add one \emph{guard agent}~$g$ and one \emph{special agent}~$s^*$.
    
    The utilities are as follows:
    \begin{description}
        \item[Element agent~$a_i$:]~$\util_{a_i}(x) = 1$ if~$x \in \{ 0, 1 \} \cup \{ 12j \mid e_i \in S_j \}$, and~$0$ otherwise. That is,~$a_i$ is happy fully isolated, alone in a clique (with only the apex agent as a neighbor), or in the cluster of any set containing~$e_i$.
        \item[Set agent~$s_j^\ell$:]~$\util_{s_j^\ell}(x) = 1$ if~$x \in \{0,1,3,12j\}$, and~$0$ otherwise. A set agent is happy fully isolated or alone in a clique, in a triplet of set agents (when her set is unselected), or in her full set-cluster of size~$12j$ (when her set is selected).
        \item[Guard agent~$g$:]~$\util_{g}(x) = 1$ if~$x \in \{0, 2\}$, and~$0$ otherwise. Besides being fully isolated, the guard is happy in a cluster of size exactly~$2$, which, when the apex is occupied, forces her to ``chase'' any agent that is alone in a clique --- the guard can jump there to form a pair.
        \item[Special agent~$s^*$:]~$\util_{s^*}(x) = 1$ if~$x = \numAgents - 1$ (recall that~$\numAgents$ denotes the total number of agents), and~$0$ otherwise. The special agent is happy only when adjacent to every other agent, which can be achieved only at the apex of the topology described below.
    \end{description}
    
    The topology~$G$ consists of~$\numAgents + 1$ disjoint cliques~$C_1,\ldots,C_{\numAgents+1}$, each of size~$\numAgents$, together with an \emph{apex} vertex~$v$ adjacent to every vertex of every clique. The role of the apex is to give the special agent a unique location achieving her peak utility, while simultaneously providing every other agent with a constant ``$+1$ neighbor'' contribution from~$s^*$. The cliques act as bins for clusters of agents, and having~$\numAgents + 1$ cliques ensures, by the pigeonhole principle, that at least one clique is empty in every allocation, giving any unhappy agent the option to jump and become a singleton. See \Cref{fig:jump:binary:NPc:topology} for an illustration.

    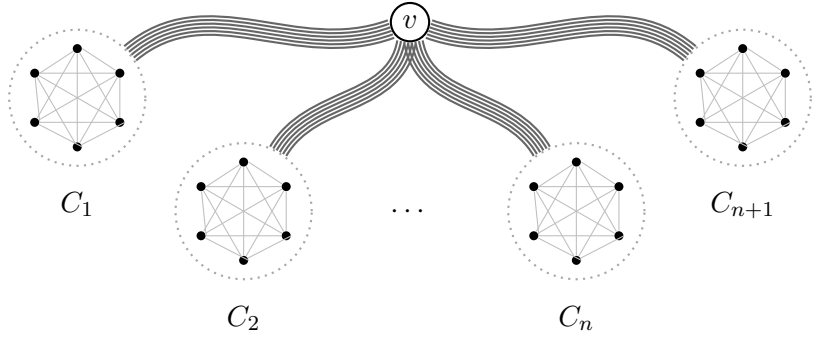
\begin{figure}[tb!]
        \centering
        \begin{tikzpicture}[
            clique/.style={draw, circle, minimum size=18mm, dotted, gray!70, thick},
            agent/.style={circle, fill=black, inner sep=1.2pt},
        ]
            \node[draw,circle,thick,inner sep=2.5pt] (v) at (0, 3) {$v$};

            \foreach \i/\xpos/\ypos/\name in {1/-4.4/2/$C_1$, 2/-2.2/0.5/$C_2$, 3/2.2/0.5/$C_n$, 5/4.4/2/$C_{\numAgents+1}$} {
                \node[clique] (C\i) at (\xpos,\ypos) {};
                \node[below=2mm of C\i] {\name};
                \foreach[count=\j] \angle in {30, 90, 150, 210, 270, 330} {
                    \node[agent] (a\i\j) at ($(C\i) + (\angle:6.5mm)$) {};
                }
            }

            \foreach \l in {1,2,3,5} {
                \foreach \i in {1,...,5} {
                    \pgfmathsetmacro \ii {\i+1};
                    \foreach \j in {\ii,...,6} {
                        \draw[gray!50] (a\l\i) -- (a\l\j);
                    }
                }
            }

            \node at (0,0.5) {$\cdots$};

            \draw[black!60,thick,scaling nfold=6] (v) to[out=200,in=40] (C1);
            \draw[black!60,thick,scaling nfold=6] (v) to[out=260,in=60] (C2);
            \draw[black!60,thick,scaling nfold=6] (v) to[out=280,in=120] (C3);
            \draw[black!60,thick,scaling nfold=6] (v) to[out=340,in=140] (C5);
        \end{tikzpicture}
        \caption{The topology~$G$ used in the reduction of \Cref{thm:jump:binary:NPc}. The apex vertex~$v$ is connected to every vertex of every clique~$C_1, \ldots, C_{\numAgents+1}$, each of size~$\numAgents$. Dotted circles represent cliques; bold dots represent vertices.}
        \label{fig:jump:binary:NPc:topology}
    \end{figure}

    Throughout the proof, we use the following observation. With~$s^*$ placed at the apex~$v$, every other agent~$i$ placed in a clique~$C_q$ has exactly~$|\alloc(\agents) \cap V(C_q)|$ allocated neighbors: all other agents in~$C_q$ (since~$C_q$ is a clique) plus~$s^*$ at the apex. We refer to~$|\alloc(\agents) \cap V(C_q)|$ as the \emph{cluster size} of~$C_q$. Note also that, by the pigeonhole principle, with~$\numAgents + 1$ cliques and at most~$\numAgents$ agents, at least one clique is always empty in any allocation.

    For correctness, suppose first that~$\mathcal{I}$ is a \Yes-instance with exact cover~$X = \{S_{j_1},\ldots,S_{j_p}\}$. Construct~$\alloc^*$ as follows: place~$s^*$ at the apex~$v$; for each~$i \in [p]$, place all~$12 j_i - 3$ set agents of~$S_{j_i}$ together with the three element agents~$\{a_\ell \mid e_\ell \in S_{j_i}\}$ in clique~$C_i$; for each unselected set~$S_j$, place its~$12j - 3$ set agents in arbitrary triplets in distinct otherwise-empty cliques; and place~$g$ alone in some otherwise-empty clique. Since the apex is adjacent to all other vertices,~$s^*$ has~$\numAgents - 1$ neighbors and utility~$1$; by the cluster-size observation, every agent in clique~$C_i$ has~$12 j_i$ neighbors and utility~$1$, every set agent in a triplet has~$3$ neighbors and utility~$1$, and~$g$ has~$1$ neighbor and utility~$0$. Every non-guard agent already attains utility~$1$, the maximum value of a binary utility, so no such agent has an improving jump. The guard would benefit only from having~$0$ or~$2$ occupied neighbors. Since~$s^*$ occupies the apex, every clique vertex has at least one occupied neighbor, so~$g$ cannot become fully isolated; and since no clique other than~$g$'s own has cluster size~$1$, no jump of~$g$ yields exactly~$2$ occupied neighbors. Hence~$\alloc^*$ is jump stable.

    Conversely, suppose~$\Gamma$ admits a jump stable allocation~$\alloc$. We establish the structure of~$\alloc$ through a sequence of claims.

    \begin{claim}\label{cl:jump:apex}
        In~$\alloc$, the apex is occupied by~$s^*$.
    \end{claim}
    \begin{claimproof}
        First, observe that no agent except~$s^*$ approves the value~$\numAgents - 1$: the largest value approved by any element or set agent is~$12j \leq 36p$, the guard approves only~$0$ and~$2$, and~$\numAgents - 1 \geq 3p + (12 \cdot 3p - 3) + 1 > 36p$. Now, suppose that the apex~$v$ is occupied by some agent~$w \neq s^*$. Since~$v$ is adjacent to every clique vertex,~$w$ has exactly~$\numAgents - 1$ occupied neighbors and, by the observation above, utility~$0$. By the pigeonhole principle, some clique is empty, and by jumping there,~$w$ becomes fully isolated with~$0$ occupied neighbors and utility~$1$; a contradiction with~$\alloc$ being jump stable. Hence, the apex is either empty or occupied by~$s^*$. Suppose the apex is empty, and let~$c$ be the cluster size of the clique hosting~$s^*$. If~$c = \numAgents$, i.e., all agents share one clique, then every agent has~$\numAgents - 1$ occupied neighbors; hence, every agent except~$s^*$ has utility~$0$ and improves by jumping to a vertex of an empty clique, where she is fully isolated and has utility~$1$. If~$c < \numAgents$, then~$s^*$ has~$c - 1 < \numAgents - 1$ occupied neighbors and utility~$0$, and she improves by jumping to the apex, where she becomes adjacent to all remaining~$\numAgents - 1$ agents. In either case,~$\alloc$ is not jump stable; a contradiction. Consequently,~$\alloc(s^*) = v$.
    \end{claimproof}

    By \Cref{cl:jump:apex} and the cluster-size observation, every non-special agent's neighbor count in~$\alloc$ equals the cluster size of her clique.

    \begin{claim}\label{cl:jump:singleton}
        In~$\alloc$, the only agent (if any) placed alone in a clique is the guard~$g$.
    \end{claim}
    \begin{claimproof}
        Suppose some non-guard agent~$a$ is alone in clique~$C_q$. Then~$g$ is in some other clique with cluster size~$c_g$, hence~$\util_g(\alloc) = 0$ unless~$c_g = 2$. If~$c_g \neq 2$,~$g$ jumps to~$C_q$, making cluster size~$2$ and gaining utility~$1$ --- improving. If~$c_g = 2$, then~$g$'s clique-mate is some agent~$b$ which, since~$s^*$ occupies the apex by \Cref{cl:jump:apex}, is an element or set agent;~$b$ has~$2$ occupied neighbors, hence utility~$0$ (no element or set agent approves~$2$), and~$b$ jumps to an empty clique (which exists by the pigeonhole principle) for cluster size~$1$ and utility~$1$ --- improving. Either way,~$\alloc$ is not jump stable, a contradiction.
    \end{claimproof}

    Next, we show that each clique contains only a certain number of agents.

    \begin{claim}\label{cl:jump:cluster-sizes}
        In~$\alloc$, every clique has cluster size in~$\{0, 1, 3\} \cup \{12j \mid j \in [3p]\}$.
    \end{claim}
    \begin{claimproof}
        Consider a clique~$C_q$ with cluster size~$c \notin \{0, 1, 3\} \cup \{12j : j \in [3p]\}$. By \Cref{cl:jump:singleton},~$C_q$ contains at least one non-guard agent~$a$ (a singleton-guard would have~$c = 1$, excluded); moreover,~$a \neq s^*$, as~$s^*$ occupies the apex by \Cref{cl:jump:apex}. Agent~$a$ has~$c$ neighbors and utility~$0$ (since~$c$ is not in~$a$'s approval set: element agents approve~$\{0, 1\} \cup \{12j\}$, set agents~$\{0, 1, 3, 12j\}$, and~$c \notin \{0, 1, 3\} \cup \{12j \mid j \in [3p]\}$). By \Cref{cl:jump:singleton}, no other clique has cluster size~$1$, so by the Pigeonhole principle, some clique is empty;~$a$ jumps there to become a singleton with utility~$1$. Improving, contradicting stability.
    \end{claimproof}

    Then, we show that ``large'' cliques have a very specific structure in terms of the agents they can contain.

    \begin{claim}\label{cl:jump:size-12j}
        Every cluster of size~$12j$ in~$\alloc$ consists only of set agents of~$S_j$, element agents~$a_\ell$ with~$e_\ell \in S_j$, and possibly the guard~$g$.
    \end{claim}
    \begin{claimproof}
        An agent in a cluster of size~$12j$ has~$12j$ occupied neighbors. The element agent~$a_i$ has utility~$1$ at~$12j$ if and only if~$e_i \in S_j$, and the set agent~$s_{j'}^\ell$ has utility~$1$ at~$12j$ if and only if~$j' = j$. Any element or set agent with utility~$0$ in this cluster could jump to an empty clique (which exists by the pigeonhole principle) for cluster size~$1$ and utility~$1$, contradicting stability. The special agent occupies the apex by \Cref{cl:jump:apex}. Hence, apart from the guard~$g$ --- who has utility~$0$ in such a cluster but may lack an improving jump --- the cluster contains only set agents of~$S_j$ and element agents~$a_i$ with~$e_i \in S_j$.
    \end{claimproof}

    Now, we are ready to show that~$\mathcal{I}$ is indeed a \Yes-instance. We set 
    \[
        X = \big\{S_j \mid \alloc \text{ contains a cluster of size } 12j \text{ for some }j\in[3p]\big\}
    \] and claim that~$X$ is an exact cover of~$\mathcal{U}$ of size~$p$. By \Cref{cl:jump:cluster-sizes}, every cluster has size~$0$,~$1$,~$3$, or~$12j$. A size-$1$ cluster contains only~$g$ by \Cref{cl:jump:singleton}, and a size-$3$ cluster cannot contain an element agent, as she would have utility~$0$ at~$3$ occupied neighbors and an improving jump to an empty clique. Hence, every element agent lies in some cluster of size~$12j$ which, by \Cref{cl:jump:size-12j}, contains only agents associated with~$S_j$ and possibly~$g$; in particular,~$e_i \in S_j$ for every element agent~$a_i$ in this cluster. Moreover, for every~$j$, there is at most one cluster of size~$12j$, as two such disjoint clusters would require at least~$2 \cdot 12j - 1$ agents associated with~$S_j$, while only~$12j$ exist. Now, we count the element agents. A cluster of size~$12j$ consists of at most~$12j - 3$ set agents of~$S_j$, at most three element agents, and possibly~$g$; since there is a single guard agent, at most one slot in at most one cluster of~$X$ is occupied by~$g$, and thus the clusters of~$X$ together contain at least~$3|X| - 1$ and at most~$3|X|$ element agents. As shown above, they contain all~$3p$ element agents, so~$3p \in \{3|X| - 1, 3|X|\}$, and since~$3|X| - 1$ is not divisible by three, we get~$3|X| = 3p$, i.e.,~$|X| = p$. Consequently, every cluster of~$X$ contains exactly three element agents, which are precisely the element agents of its set~$S_j$. Since each element agent lies in exactly one cluster, the sets in~$X$ are pairwise disjoint as subsets of~$\mathcal{U}$ and together cover all~$3p$ elements; that is,~$X$ is an exact cover of size~$p$.

    The reduction can be clearly performed in polynomial time, completing the \NPhness part of the proof. As membership in \NP is trivial: given an allocation~$\alloc$, one can check in polynomial time that~$\util_i(\alloc) \geq \util_i(\alloc^{i\mapsto v})$ for every agent~$i\in\agents$ and every empty vertex~$v\in V(G)$, we conclude that the problem under consideration is indeed \NPc.
\end{proof}

\subsection{Local Jumps}\label{sec:localJumps}

We now focus on the case where agents can only jump a certain distance away from their location. Clearly, for utility functions that are all monotone of the same type, such an allocation always exists due to \Cref{thm:jump:monotone}. Our first observation is a strengthening of the non-existence of jump-stable allocation when agents can only jump to vertices that are neighbors of their current location.

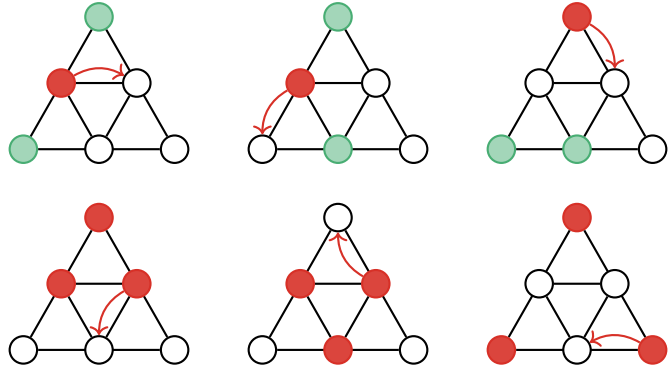
\begin{figure}[bt!]
    \centering
    \begin{tikzpicture}
        \node[draw,circle,thick,cbGood!80,fill=cbGood!40] (A) at (0,0) {};
        \node[draw,circle,thick,cbGood!80,fill=cbGood!40] (B) at (-1,-1.75) {};
        \node[draw,circle,thick] (C) at (1,-1.75) {};

        \draw[thick] (A) -- (B) node[draw,midway,circle,cbHard,fill=cbHard!90] (d) {};
        \draw[thick] (A) -- (C) node[draw,midway,circle,fill=white] (e) {};
        \draw[thick] (B) -- (C) node[draw,midway,circle,fill=white] (f) {};
        \draw[thick] (d) -- (e) -- (f) -- (d); 
        \draw[cbHard,thick,->] (d) to[out=30,in=150] (e);
    \end{tikzpicture}\qquad
    \begin{tikzpicture}
        \node[draw,circle,thick,cbGood!80,fill=cbGood!40] (A) at (0,0) {};
        \node[draw,circle,thick] (B) at (-1,-1.75) {};
        \node[draw,circle,thick] (C) at (1,-1.75) {};

        \draw[thick] (A) -- (B) node[draw,midway,circle,cbHard,fill=cbHard!90] (d) {};
        \draw[thick] (A) -- (C) node[draw,midway,circle,fill=white] (e) {};
        \draw[thick] (B) -- (C) node[draw,midway,circle,cbGood!80,fill=cbGood!40] (f) {};
        \draw[thick] (d) -- (e) -- (f) -- (d);
        \draw[cbHard,thick,->] (d) to[out=210,in=90] (B);
    \end{tikzpicture}\qquad
    \begin{tikzpicture}
        \node[draw,circle,thick,cbHard,fill=cbHard!90] (A) at (0,0) {};
        \node[draw,circle,thick,cbGood!80,fill=cbGood!40] (B) at (-1,-1.75) {};
        \node[draw,circle,thick] (C) at (1,-1.75) {};

        \draw[thick] (A) -- (B) node[draw,midway,circle,fill=white] (d) {};
        \draw[thick] (A) -- (C) node[draw,midway,circle,fill=white] (e) {};
        \draw[thick] (B) -- (C) node[draw,midway,circle,cbGood!80,fill=cbGood!40] (f) {};
        \draw[thick] (d) -- (e) -- (f) -- (d);
        \draw[cbHard,thick,->] (A) to[out=-30,in=90] (e);
    \end{tikzpicture}

    \qquad

    \begin{tikzpicture}
        \node[draw,circle,thick,cbHard,fill=cbHard!90] (A) at (0,0) {};
        \node[draw,circle,thick] (B) at (-1,-1.75) {};
        \node[draw,circle,thick] (C) at (1,-1.75) {};

        \draw[thick] (A) -- (B) node[draw,midway,circle,cbHard,fill=cbHard!90] (d) {};
        \draw[thick] (A) -- (C) node[draw,midway,circle,cbHard,fill=cbHard!90] (e) {};
        \draw[thick] (B) -- (C) node[draw,midway,circle,fill=white] (f) {};
        \draw[thick] (d) -- (e) -- (f) -- (d); 
        \draw[cbHard,thick,->] (e) to[out=210,in=90] (f);
    \end{tikzpicture}\qquad
    \begin{tikzpicture}
        \node[draw,circle,thick] (A) at (0,0) {};
        \node[draw,circle,thick] (B) at (-1,-1.75) {};
        \node[draw,circle,thick] (C) at (1,-1.75) {};

        \draw[thick] (A) -- (B) node[draw,midway,circle,cbHard,fill=cbHard!90] (d) {};
        \draw[thick] (A) -- (C) node[draw,midway,circle,cbHard,fill=cbHard!90] (e) {};
        \draw[thick] (B) -- (C) node[draw,midway,circle,cbHard,fill=cbHard!90] (f) {};
        \draw[thick] (d) -- (e) -- (f) -- (d); 
        \draw[cbHard,thick,->] (e) to[out=150,in=270] (A);
    \end{tikzpicture}\qquad
    \begin{tikzpicture}
        \node[draw,circle,thick,cbHard,fill=cbHard!90] (A) at (0,0) {};
        \node[draw,circle,thick,cbHard,fill=cbHard!90] (B) at (-1,-1.75) {};
        \node[draw,circle,thick,cbHard,fill=cbHard!90] (C) at (1,-1.75) {};

        \draw[thick] (A) -- (B) node[draw,midway,circle,fill=white] (d) {};
        \draw[thick] (A) -- (C) node[draw,midway,circle,fill=white] (e) {};
        \draw[thick] (B) -- (C) node[draw,midway,circle,fill=white] (f) {};
        \draw[thick] (d) -- (e) -- (f) -- (d);
        \draw[cbHard,thick,->] (C) to[out=150,in=30] (f);
    \end{tikzpicture}
    
    \caption{All possible non-isomorphic allocations in the instance used to prove~\Cref{thm:1-jump-counterexample}. Recall that our agents have identical preferences such that utility is~$1$ if and only if they have exactly one neighbor. We use green (light gray) to highlight agents with utility~$1$ in an allocation, and red (dark gray) to highlight agents with utility~$0$. In each case, we depict one possible improving~$1$-jump deviation.}
    \label{fig:1-jump-counterexample}
\end{figure}

We begin by showing that~$1$-jump stable allocations may fail to exist in general.

\begin{theorem}
    \label{thm:1-jump-counterexample}
    A~$1$-jump stable allocation may fail to exist, even if there are three agents with identical binary/single-peaked utility functions.
\end{theorem}
\begin{proof}
    Consider the following game with three identical agents, where the utility of an agent is~$1$ if she has exactly one neighbor, and~$0$ otherwise.
    
    No matter how we allocate the three agents on the vertices of the topology, there will always exist one agent that: 
    a) has either two or zero occupied neighboring vertices, and 
    b) has one empty neighboring vertex (i.e., a valid~$1$-jump move) with exactly one neighbor. 
    There are four different cases for the configurations of the locations of the agents (see \Cref{fig:1-jump-counterexample}): 
    \begin{itemize}
        \item they form a path of length~$2$ (top left and top center configuration), where the agent in the middle always has a profitable~$1$-jump; 
        \item they form an edge plus an isolated vertex (top right configuration), where the isolated agent can always perform a~$1$-jump to a vertex that is adjacent to exactly one occupied vertex;
        \item they form a triangle (bottom left and bottom center configuration), where there is always at least one agent that has a profitable~$1$-jump; 
        \item they form the unique independent set of the graph (bottom right), and each player can~$1$-jump to any neighboring vertex and be exactly adjacent to one other agent.
    \end{itemize}
    Thus, in any configuration, one agent will have an incentive to change her location, and no~$1$-jump stable allocation exists.
\end{proof}

We next observe that the hardness of deciding the existence of jump stable allocations holds, even if we restrict the jump to the distance at most~$2$. 

\begin{theorem}
    \label{thm:twoJump:binary:NPc}
    Given a game~$\Gamma$, it is \NPc to decide whether it admits a~$2$-jump stable allocation.
\end{theorem}
\begin{proof}
    Even though it was not stated formally, it is easy to see that the diameter of the topology used to prove \Cref{thm:jump:binary:NPc} is~$2$---two distinct vertices~$u$ and~$w$ are either direct neighbors, or share the apex vertex~$v$ as a common neighbor. Hence, we immediately obtain \NPcness for~$2$-jump stability as well. 
\end{proof}

The topology constructed in the proof of \Cref{thm:twoJump:binary:NPc} is of diameter two. In contrast, when the topology is of diameter one (i.e.,~$G$ is a complete graph), we have trivial existence even for unrestricted utility functions.

\begin{observation}
    If~$G$ has diameter~$1$, then every allocation is jump stable.
\end{observation}

To see this, observe that in a complete graph, every agent has exactly~$\numAgents - 1$ occupied neighboring vertices in every allocation; hence, no jump can change the utility of any agent.

We next focus on acyclic graphs and explore the existence of~$1$-jump-stable and~$2$-jump-stable allocations. 

\begin{theorem}
\label{thm:oneJump:binary:trees}
    If the topology is a tree and agents have binary utility functions, a~$1$-jump stable allocation is guaranteed to exist and can be found in polynomial time. On the other hand, a~$2$-jump stable allocation may fail to exist, even if the topology is a simple path, and the utility functions are identical and binary/single-peaked.
\end{theorem}
\begin{proof}
    Assume that we are given a tree graph~$T$ and we have a set of agents with binary utilities. We root~$T$ at an arbitrary vertex; the notions of parent, child, and height below refer to this rooting. In order to construct a 1-jump stable allocation, we follow the procedure described next. 
    \begin{itemize}
        \item We first order the vertices of~$T$ with respect to the height of their sub-tree, i.e., leaves have height zero and every other vertex has height equal to the maximum height of its children plus 1.
        \item We visit the vertices according to this order, and for each vertex~$v$, we check the following.
        \begin{itemize}
            \item If the number of remaining unvisited vertices (including~$v$) is equal to the number of unallocated agents, then we allocate all unallocated agents to the unvisited vertices in an arbitrary way.
            \item Otherwise, we check if there is an unallocated agent that gets utility 1 when allocated to~$v$ when the children of~$v$ have a fixed allocation and assuming that the parent of~$v$ will remain empty. If there is such an agent, we allocate her to~$v$. Otherwise, we leave~$v$ empty, and we move to the next vertex of the ordering.
        \end{itemize}
    \end{itemize}
    We will prove that indeed the constructed allocation is 1-jump stable. Clearly, if an agent has utility 1, it has no profitable 1-jump. Hence, in what follows, we focus on agents with utility 0 and an empty neighboring vertex.  
    
    We begin by considering the agents that were allocated greedily when their vertex was visited. Assume that agent~$i$ was allocated to vertex~$v$ when~$v$ was visited. In order for agent~$i$ to have utility 0, it means that the parent vertex of~$v$ has been allocated to a different agent later in the process; otherwise~$i$ would have had utility 1, since she obtained utility 1 at~$v$ under the fixed allocation of the children of~$v$ and an empty parent, and the allocation of the children of~$v$ never changes afterwards. Observe that agent~$i$ cannot 1-jump to the parent vertex of~$v$, since it is occupied. Thus, if there is an empty vertex in the neighborhood of~$v$, it is a child of~$v$. We claim that she has no incentives to jump to such a vertex, which we denote~$z$. 
    To see why this is the case, observe that our procedure has visited~$z$ before~$v$, and agent~$i$ was not allocated to~$z$. This is because the agent would get utility 0 given the allocated agents on the children of~$z$ while having its parent-vertex, i.e.,~$v$, empty. Thus, the utility of agent~$i$ after the 1-jump from~$v$ to~$z$ is still 0.
    
    Finally, we consider the agents that were allocated during the last step of the procedure, i.e., to the unvisited vertices. Observe that the procedure guarantees that all parents of unvisited vertices have an allocated agent. Thus, the only possible 1-jumps are towards child vertices that have been visited in the past. Again, we claim that no agent has incentives to 1-jump to such a vertex, following the reasoning above: if an agent could get utility 1 from jumping to a visited vertex~$v$, then she would have been allocated to this vertex when our procedure was considering~$v$. Hence, no agent overall has a profitable 1-jump, and the constructed allocation is 1-jump stable.
    
    For the second claim of the theorem, consider a path~$v_1 v_2 v_3 v_4 v_5$ on five vertices and three agents with identical binary utilities. Every agent has utility 1 if she is next to exactly one occupied vertex and 0 otherwise. We claim that in every configuration, at least one agent has a profitable 2-jump.
    \begin{itemize}
        \item If the allocated vertices form an independent set, then all three can move to a neighboring vertex and be next to an occupied vertex, and increase their utility from 0 to 1.
        \item If the allocated vertices form a path, then the agent at the middle of the path always has a 2-jump that places her next to exactly one occupied vertex, i.e., the occupied vertices form an edge and an independent vertex.
        \item If the allocated vertices form an edge and an independent vertex, then, up to symmetry, the configuration is~$\{v_1, v_2, v_4\}$,~$\{v_1, v_2, v_5\}$, or~$\{v_2, v_3, v_5\}$, and the agent on the independent vertex has a profitable 2-jump to~$v_3$,~$v_3$, or~$v_4$, respectively: in each case, the target vertex is within distance two of her current location and, once she leaves it, is adjacent to exactly one occupied vertex, so her utility increases from 0 to 1.
    \end{itemize}
    That is, in no configuration is there a stable solution, finishing the proof.
\end{proof}
    
\section{Social Welfare and Pareto Optimality}\label{sec:efficiency}

Now, we shift our focus from stability to efficiency. In this section, we will focus on finding optimal allocations with respect to the efficiency notions introduced in \Cref{sec:model}. We start with the case of utilitarian social welfare, egalitarian social welfare, and Nash social welfare, for which we show that, given a bound, it is computationally intractable to decide whether an allocation with welfare at least this bound exists or not. The computational intractability persists, even if the number of agents to allocate is small.

\begin{theorem}
    \label{thm:USW:NPc}\label{thm:ESW:NPc}\label{thm:NSW:NPc}
    \label{thm:USW:numAgents:Wh}\label{thm:ESW:numAgents:Wh}\label{thm:NSW:numAgents:Wh}
    It is \NPc and \Wc when parameterized by the number of agents to decide whether there exists an allocation with utilitarian social welfare at least~$\zeta$, even if the utilities are binary, linearly increasing, or linearly decreasing. The same holds for egalitarian and Nash social welfare.
\end{theorem}
\begin{proof}
    Membership in \NP is trivial: given an allocation~$\alloc$ as a certificate, one can compute~$\mathcal{SW}(\alloc)$, where~$\mathcal{SW}$ is an arbitrary social welfare function, in polynomial time and check whether~$\mathcal{SW}(\alloc) \geq \zeta$.

    To show \NPhness (and \Whness), we show a reduction from the \probName{$k$-Clique} problem. Here, we are given a graph~$H$ and an integer~$k\in\N$, and the goal is to decide if there is a set~$K\subseteq V(H)$ of size~$k$ such that~$K$ induces a complete graph in~$H$. The problem is very well known to be \NPh~\citep{Karp1972} and \Wc when parameterized by the solution size~$k$~\citep{DowneyF1995}.

    Given an instance~$\mathcal{I}=(H,k)$ of the \probName{$k$-Clique} problem, we construct an equivalent game~$\Gamma$ as follows. The topology~$G$ is identical to the graph~$H$, and we set~$\agents = \{a_1,\ldots,a_k\}$. The utility function of every agent~$i\in\agents$ is defined as
    \[
        \util_i(x) = \begin{cases}
            1 & \text{if } x = k - 1 \text{ and}\\
            0 & \text{otherwise.}
        \end{cases}
    \]
    Clearly, the utilities are identical and binary (we will show how to tweak the reduction for linearly increasing and linearly decreasing utilities later in the proof). To finalize the construction, we set~$\zeta = k$.

    For correctness, assume that~$\mathcal{I}$ is a \Yes-instance and let~$K$ be the clique of size~$k$ in~$H$. Let~$\alloc$ be an allocation mapping the agents arbitrarily to the vertices of~$K$. Since~$K$ is a clique, every agent is adjacent to~$k-1$ other agents, and by the definition of the utility function, has utility~$1$. Therefore, we have~$\USW(\alloc) = k$, which is exactly the required bound. Conversely, if there is an allocation~$\alloc$ with social welfare at least~$\zeta = k$, then every agent must have all the remaining agents in her neighborhood. Hence,~$\alloc$ maps the agents to a clique of size~$k$ in~$H$, finishing the proof for binary utilities.

    \proofparagraph{Linearly increasing utilities.} For linearly increasing utilities, we set~$\util_i(x) = x$ for every~$i$ and~$\zeta = k\cdot(k-1)$. This bound can be again achieved if and only if each agent is the neighbor of all other agents.

    \proofparagraph{Linearly decreasing utilities.} For linearly decreasing utilities, we set~$\util_i(x) = k - 1 - x$,~$G = \bar{H}$ (i.e., the complement of graph~$H$), and~$\zeta = k\cdot(k-1)$ (as in the case of linearly increasing utilities). Now, an allocation~$\alloc$ achieves the bound~$\zeta$ if and only if the agents have no neighbors, which in turn means that they are allocated to a clique in~$H$.

    \proofparagraph{ESW and NSW.} For ESW, we set~$\zeta = 1$ for binary utilities and~$\zeta = k-1$ for linearly increasing and linearly decreasing utilities. This corresponds to the case when every agent has all other agents as her neighbors (or none in the case of linearly decreasing), which is exactly the crucial property for the correctness of the reduction. For NSW, we again set~$\zeta = 1$ for binary utilities, and use~$\zeta = (k-1)^k$ for the remaining utility functions.

    \proofparagraph{Membership in \W.} We give a parameterized reduction to the \probName{Short Nondeterministic Turing Machine Acceptance} problem, which asks whether a given single-tape nondeterministic Turing machine~$M$ accepts its input within~$k$ steps, and which is \Wc when parameterized by~$k$~\citep{Cesati2003}. Crucially, only the number of steps is bounded by the parameter: the state set and the alphabet of~$M$ may be of size polynomial in the input, and both are produced by the reduction in polynomial time; in particular, any polynomial-time-precomputable table can be hard-wired into the transition function of~$M$.

    The main obstacle is that the machine must compare welfare values against~$\zeta$ within a number of steps depending only on~$\numAgents$, while the numbers involved may be exponentially large in the input size. We overcome it using the weight-reduction theorem of \citet{FrankT1987}: given a vector~$w \in \mathbb{Q}^d$ and an integer~$N'$, one can compute in polynomial time a vector~$\bar{w} \in \mathbb{Z}^d$ with~$\|\bar{w}\|_\infty \leq 2^{\Oh{d^3}} \cdot (N')^{\Oh{d^2}}$ such that~$\operatorname{sign}(w \cdot b) = \operatorname{sign}(\bar{w} \cdot b)$ for every~$b \in \mathbb{Z}^d$ with~$\|b\|_1 \leq N' - 1$.

    For USW, let~$w \in \mathbb{Z}^{\numAgents^2 + 1}$ consist of all utility values~$\util_i(x)$,~$i \in [\numAgents]$,~$x \in [\numAgents - 1]_0$, followed by~$\zeta$. For any allocation in which agent~$i$ has~$x_i$ occupied neighboring vertices, the comparison~$\sum_{i} \util_i(x_i) \geq \zeta$ asks for the sign of~$w \cdot b$, where~$b$ has~$\numAgents$ entries equal to~$1$ and one entry equal to~$-1$; hence,~$\|b\|_1 = \numAgents + 1$. Applying the theorem with~$d = \numAgents^2 + 1$ and~$N' = \numAgents + 2$ yields replacement values~$\bar{\util}_i(x)$ and~$\bar{\zeta}$ inducing the same comparisons, whose bit-length is in~$B \in \Oh{\numAgents^6}$.

    For NSW, we first turn the multiplicative comparison into an additive one. If~$\zeta = 0$, the instance is trivially a \Yes-instance, so assume~$\zeta \geq 1$; moreover, allocations in which some agent has utility~$0$ have Nash welfare~$0 < \zeta$ and are recognized directly via a hard-wired table. Otherwise, since the product and~$\zeta$ are integers,~$\prod_i \util_i(x_i) \geq \zeta$ holds if and only if~$\sum_i \log_2 \util_i(x_i) - \log_2(\zeta - \nicefrac{1}{2}) > 0$. Let~$\ell$ be the maximum bit-length of the numbers in the input. Then, both the product and~$\zeta - \nicefrac{1}{2}$ are at most~$2^{\numAgents \ell}$ and differ by at least~$\nicefrac{1}{2}$, so the absolute value of the linear form above is at least~$2^{-(\numAgents\ell + 2)}$. We compute, in polynomial time, rational approximations of the~$\numAgents^2 + 1$ logarithms with additive error at most~$2^{-(\numAgents\ell + 4)}/(\numAgents + 1)$; the sign of the approximate linear form then coincides with the sign of the exact one for every allocation, and we apply the Frank--Tardos reduction to the vector of approximations exactly as in the case of USW. For ESW, no weight reduction is necessary: the reduction precomputes the Boolean table with entries~$T(i,x) = 1$ if and only if~$\util_i(x) \geq \zeta$.

    The machine~$M$, over an alphabet containing one symbol per vertex of~$G$ together with the binary digits, operates as follows. First,~$M$ nondeterministically writes~$\numAgents$ vertex symbols~$v_1, \ldots, v_\numAgents$ on the tape, interpreted as the allocation~$\alloc(i) = v_i$. Then, it deterministically verifies that the guessed vertices are pairwise distinct and computes, for each~$i \in [\numAgents]$, the number~$x_i = |\{ j \neq i \mid v_j \in N_G(v_i)\}|$ of occupied neighboring vertices; the adjacency relation of~$G$ is hard-wired in the transition function, the counters take values in~$[\numAgents - 1]_0$ and are stored in the state, and each pairwise check costs~$\Oh{\numAgents}$ head movements, so this phase takes~$\Oh{\numAgents^3}$ steps. For ESW,~$M$ then simply checks that~$T(i, x_i) = 1$ for every~$i \in [\numAgents]$. For USW and NSW,~$M$ maintains the running sum of the reduced values associated with the pairs~$(i, x_i)$ in binary on a scratch portion of the tape: every involved number has bit-length~$\Oh{B}$, the bits of the reduced values are emitted by the transition function (with the state storing the pair~$(i, x_i)$, the current bit position, and the carry), and hence each of the~$\numAgents$ additions, as well as the final comparison against the reduced threshold, takes~$\Oh{B + \numAgents}$ steps.

    Overall, some computation path of~$M$ accepts within~$k \in \Oh{\numAgents^3 + \numAgents \cdot B} \subseteq \Oh{\numAgents^7}$ steps if and only if~$\Gamma$ admits an allocation with the respective social welfare at least~$\zeta$. Since~$k$ depends only on~$\numAgents$ and the whole construction runs in polynomial time, this is a correct parameterized reduction, concluding the proof.
\end{proof}

Next, we turn our attention to Pareto optimality. First, observe that a Pareto optimal allocation always exists in our model, regardless of the utility profile. The set of allocations is finite and non-empty, and Pareto domination is a strict partial order on this set. Every finite non-empty strict partial order has a maximal element, and a maximal element under Pareto domination is by definition a Pareto optimal allocation. 

\begin{observation}
    \label{thm:PO:alwaysExists}
    Every game~$\Gamma$ admits at least one Pareto optimal allocation.
\end{observation}

Hence, we focus on the complexity of verifying and finding such allocations rather than their existence. Unfortunately, the result here is very negative: already verifying whether an allocation~$\alloc$ is Pareto optimal is computationally hard, even if the number of agents is small.

\begin{theorem}
    \label{thm:PO:verify:coNPh}\label{thm:PO:verify:numAgents:coWh}
    It is \coNPh and \coWh when parameterized by the number of agents to verify whether an allocation~$\alloc$ is Pareto optimal, even if all utility functions are identical and binary, linearly increasing, or linearly decreasing.
\end{theorem}
\begin{proof}
    We again reduce from the \probName{$k$-Clique} problem, i.e., given a graph~$H$ and an integer~$k$, decide if there is~$K\subseteq V(H)$,~$|K| = k$, such that~$H[K]$ is a complete graph. This problem is known to be \NPc~\citep{Karp1972} and \Wc when parameterized by the solution size~$k$~\citep{DowneyF1995}.

    Given~$\mathcal{I} = (H,k)$, we construct an equivalent game~$\Gamma$ as follows. We start with the description of the topology~$G$. It consists of the disjoint union of~$H$ and a graph~$G'$, which is a complete graph on~$k$ vertices with one edge removed. There are exactly~$k$ agents with identical binary utilities (we show how to tweak the reduction for linearly increasing and linearly decreasing utilities later in the proof) defined as follows:
    \[
        \util_i( x ) = \begin{cases}
            1 & \text{if } x = k-1 \text{ and}\\
            0 & \text{otherwise.}
        \end{cases}
    \]
    Finally, we ask whether an allocation~$\alloc$ that maps the agents bijectively to the vertices of~$G'$ is Pareto optimal. Since~$G'$ is~$K_k$ minus one edge, two vertices of~$G'$ have degree~$k-2$ and the remaining~$k-2$ vertices have degree~$k-1$; hence, in~$\alloc$, exactly two agents (those mapped to the endpoints of the removed edge) have utility~$0$, and every other agent has utility~$1$.

    For the correctness, assume first that~$\mathcal{I}$ is a \Yes-instance, that is, it contains a clique~$K$ of size~$k$. Then, in an allocation~$\alloc'$ mapping all agents to vertices of~$K$, the utility of every agent is exactly~$1$, as each of them has exactly~$k-1$ direct neighbors. Therefore,~$\alloc'$ Pareto dominates~$\alloc$, and we obtain that~$\alloc$ is not Pareto optimal.

    Conversely, suppose that~$\alloc$ is not Pareto optimal, and let~$\alloc'$ be an allocation that Pareto dominates~$\alloc$. Since no agent may be worse off, the~$k-2$ agents with utility~$1$ in~$\alloc$ also have utility~$1$ in~$\alloc'$, and since some agent must be strictly better off, at least one of the two agents with utility~$0$ in~$\alloc$ has utility~$1$ in~$\alloc'$. An agent has utility~$1$ if and only if all remaining~$k-1$ agents are her neighbors; hence, at least~$k-1$ of the~$k$ vertices occupied under~$\alloc'$ are adjacent to all other occupied vertices. As at most one occupied vertex lies outside this set, every pair of occupied vertices contains a vertex adjacent to all occupied vertices, so the~$k$ occupied vertices are pairwise adjacent and form a clique of size~$k$ in~$G$. Since~$G$ is the disjoint union of~$H$ and~$G'$, no clique contains vertices of both parts, and~$G'$ itself has only~$k$ vertices and is not complete; therefore, the clique lies entirely in~$H$, and~$\mathcal{I}$ is a \Yes-instance, finishing the correctness.

    The reduction runs in polynomial time. Moreover, we used~$\numAgents = k$, so the reduction is a parameterized reduction as well.

    \proofparagraph{Linearly increasing utilities.} For linearly increasing utilities, the construction is the same. The only difference is that the utility function is~$\util_i(x) = x$ for every~$i\in\agents$. Again, all agents except for two have the maximum possible utility~$k-1$, and the two remaining agents have utility~$k-2$. If~$H$ contains a clique of size~$k$, mapping the agents to its vertices gives every agent utility~$k-1$ and Pareto dominates~$\alloc$. Conversely, a Pareto-dominating allocation must keep the~$k-2$ maximum agents at utility~$k-1$ and raise at least one of the other two to~$k-1$, and the same argument as in the binary case shows that the occupied vertices form a clique of size~$k$ in~$H$.

    \proofparagraph{Linearly decreasing utilities.} For linearly decreasing utilities, we modify the construction in two ways. First, we replace the topology with its complement; that is, the new topology is~$\bar{G}$, the complement of the disjoint union~$H \sqcup G'$. Second, we set the utility of every agent~$i\in\agents$ to~$\util_i(x) = k - 1 - x$. The allocation~$\alloc$ still maps agents bijectively to the vertices of~$G'$. Observe that in~$\bar{G}$, two vertices of~$G'$ are adjacent if and only if they were non-adjacent in~$G'$; since~$G'$ is~$K_k$ minus one edge, the only adjacent pair of~$G'$-vertices in~$\bar{G}$ is the pair of endpoints of the removed edge. Therefore, in~$\alloc$, exactly two agents are mutual neighbors with utility~$\util_i(1) = k-2$, while every other agent has utility~$\util_i(0) = k-1$, the maximum value. In a Pareto-dominating allocation~$\alloc'$, the~$k-2$ agents with utility~$k-1$ must retain it, and at least one of the two agents with utility~$k-2$ must improve to~$k-1$; hence, at least~$k-1$ occupied vertices have no occupied neighbors in~$\bar{G}$. The remaining occupied vertex then has no occupied neighbors either, as every other occupied vertex is non-adjacent to all occupied vertices. Therefore,~$\alloc'(\agents)$ is an independent set of size~$k$ in~$\bar{G}$, i.e., a clique of size~$k$ in~$H \sqcup G'$. Since~$G'$ is not itself a~$k$-clique, any such clique must lie entirely in~$H$. Hence, a Pareto-dominating allocation exists if and only if~$H$ contains a clique of size~$k$.
\end{proof}

Since the reduction above is a parameterized reduction from \probName{$k$-Clique} with~$\numAgents = k$, the ETH-based lower bound for \probName{$k$-Clique} directly translates to the verification problem as well.

\begin{corollary}\label{cor:PO:verify:ETH}
    Unless ETH fails, there is no algorithm that verifies whether a given allocation~$\alloc$ is Pareto optimal in~$f(\numAgents)\cdot|V(G)|^{\oh{\numAgents}}$ time for any computable function~$f$. Consequently, the brute-force verification, which compares~$\alloc$ against all~$|V(G)|^{\Oh{\numAgents}}$ allocations, is asymptotically optimal.
\end{corollary}
\begin{proof}
    The reduction in the proof of \Cref{thm:PO:verify:coNPh} constructs, from an instance~$(H,k)$ of \probName{$k$-Clique}, a game with~$\numAgents = k$ agents and an allocation~$\alloc$ such that~$H$ contains a clique of size~$k$ if and only if~$\alloc$ is not Pareto optimal. Suppose that an algorithm~$\mathscr{A}$ verifies Pareto optimality in~$f(\numAgents)\cdot|V(G)|^{\oh{\numAgents}}$ time. Then, given an instance~$(H,k)$ of \probName{$k$-Clique}, we construct the game and the allocation in polynomial time, run~$\mathscr{A}$, and report that a clique exists if and only if~$\mathscr{A}$ rejects. As~$|V(G)| = |V(H)| + k$ and~$\numAgents = k$, this decides \probName{$k$-Clique} in~$g(k)\cdot|V(H)|^{\oh{k}}$ time for a computable function~$g$, which contradicts ETH~\citep{ChenHKX2006}.
\end{proof}

We conclude this section by presenting \XP algorithms for parameterization by the number of agents for all efficiency notions we are interested in, complementing the above hardness results. The algorithms employ a simple brute-force approach, trying all possible allocations. However, interestingly enough, such a naïve enumeration is, under ETH, asymptotically optimal.

\begin{theorem}
    \label{thm:USW:numAgents:XP}\label{thm:ESW:numAgents:XP}
    \label{thm:NSW:numAgents:XP}\label{thm:PO:numAgents:XP}
    There is an algorithm running in~$|V(G)|^\Oh{\numAgents}$ time that, given a game~$\Gamma$, finds an allocation~$\alloc$ that (i) maximizes USW, (ii) maximizes ESW, (iii) maximizes NSW, or (iv) is Pareto optimal. Moreover, unless ETH fails, there is no algorithm solving problems (i)-(iii) in time~$f(\numAgents)\cdot |V(G)|^\oh{\numAgents}$ for any computable function~$f$. 
\end{theorem}
\begin{proof}
    For problems (i)-(iii), the algorithm enumerates all possible allocations~$\alloc$, computes the value of the respective notion of social welfare, and keeps the allocation with the maximum value. As there are~$\numAgents$ agents, there are~$|V(G)|^\Oh{\numAgents}$ different allocations, and the social welfare for each allocation can be computed in polynomial time. For (iv), we simply return the allocation maximizing USW, as it is necessarily Pareto optimal: if not, then there is an allocation~$\alloc'$, where no agent is worse off and at least one agent is better off, so the overall utility is strictly higher than in~$\alloc$; a contradiction with~$\alloc$ being the max USW allocation.

    For the running time lower-bound, it is known that \probName{$k$-Clique} cannot be solved in~$g(k)\cdot |V(H)|^\oh{k}$ time for any computable function~$g$, unless the ETH is false~\citep{ChenHKX2006}. Assume that there is an algorithm~$\mathscr{A}$ solving the problem (i) in~$f(\numAgents)\cdot |V(G)|^\oh{\numAgents}$ time (for problems (ii) and (iii), the argument is identical). Then, given an instance~$\mathcal{I}$ of \probName{$k$-Clique}, we can use the reduction from \Cref{thm:USW:numAgents:Wh} to obtain an equivalent game~$\Gamma$, solve~$\Gamma$ using~$\mathscr{A}$, and return the same response for~$\mathcal{I}$. The result is clearly correct as~$\mathcal{I}$ and~$\Gamma$ are equivalent. Moreover, as the reduction can be done in polynomial time, we decided~$\mathcal{I}$ in time~$|V(H)|^\Oh{1} + f(\numAgents)\cdot |V(G)|^\oh{\numAgents} = f(\numAgents)\cdot |V(G)|^\oh{\numAgents} = g(k)\cdot |V(H)|^\oh{k}$, which contradicts ETH. Hence, it is unlikely that such an algorithm~$\mathscr{A}$ for the problem (i) exists.
\end{proof}

\section{Conclusions}\label{sec:conclusions}

As our results indicate, the existence and tractability of the underlying computational problem change drastically based on the combination of stability notion, utility function, and graph topology. One question that remains open is to provide a complete dichotomy for~$1$-jump-stability on trees. Our algorithm for binary utilities cannot immediately work for other utility functions, such as single-peaked utilities, although there might exist a sophisticated modification that makes it work. 

A different direction is to consider other stability notions, with contractual-based stability notions being prime candidates~\citep{BogomolnaiaJ2002,SungD2007,GairingS2019,BrandtBT2024}.
This family is a different refinement of jump-stability, where a jump of an agent must be ``approved'' either by her current neighbors (out-contractual), or by her potential new neighbors after the jump (in-contractual), or both (in-out-contractual). Here, there is ample room for choosing what ``approval'' actually means: it can be (weak) majority, veto, or a welfare-based function. Although a couple of our results can be extended to a subset of cases, we believe there are novel algorithmic insights that can yield positive results for these notions.

\section*{Acknowledgements} 
This research was supported in part by the National Science Centre, Poland, grant number UMO-2025/58/A/ST6/00371, by the EPSRC grant EP/X039862/1, and co-funded by the European Union under the project Robotics and Advanced Industrial Production (reg. no. CZ.02.01.01/00/22\_008/0004590).

\bibliographystyle{plainnat}
\bibliography{references}

\end{document}